%% file: main.tex
\documentclass{elsarticle}

\usepackage{amssymb,amsmath}
\usepackage{graphicx,algorithmic,algorithm}
\usepackage{enumerate}
\usepackage{tikz}
\usetikzlibrary{shapes}
\usepackage[T1]{fontenc}

 \usepackage{xcolor}
 \usepackage{mathtools}
\usepackage{microtype}
\usepackage{amsfonts}
\usepackage{comment}
\usepackage[english]{babel}
\usepackage{mathrsfs}
\usepackage[absolute]{textpos}

\usepackage{fontaxes}
\usepackage{trimspaces}
\usepackage{nccfoots}
\usepackage{setspace}

\usepackage{enumitem}
\usepackage{todonotes}

\definecolor{blue}{rgb}{0.1,0.2,0.5}
\definecolor{brown}{rgb}{0.6,0.6,0.2}
\usepackage[ocgcolorlinks, linkcolor={blue}, citecolor={brown}]{hyperref}

\usepackage[amsmath,thmmarks,hyperref]{ntheorem}
\usepackage{cleveref}

\crefformat{page}{#2page~#1#3}%
\Crefformat{page}{#2Page~#1#3}%
\crefformat{equation}{#2(#1)#3}%
\Crefformat{equation}{#2(#1)#3}%
\crefformat{figure}{#2Figure~#1#3}%
\Crefformat{figure}{#2Figure~#1#3}%
\crefformat{section}{#2Section~#1#3}
\Crefformat{section}{#2Section~#1#3}
\crefformat{chapter}{#2Chapter~#1#3}
\Crefformat{chapter}{#2Chapter~#1#3}
\crefformat{chapter*}{#2Chapter~#1#3}
\Crefformat{chapter*}{#2Chapter~#1#3}
\crefformat{part}{#2Part~#1#3}
\Crefformat{part}{#2Part~#1#3}
\crefformat{enumi}{#2(#1)#3}
\Crefformat{enumi}{#2(#1)#3}

\usepackage{enumerate}

\usepackage{latexsym}

\theoremnumbering{arabic}
\theoremstyle{plain}
\theoremsymbol{}
\theorembodyfont{\itshape}
\theoremheaderfont{\normalfont\bfseries}
\theoremseparator{.}
\theorempreskip{4pt}
\theorempostskip{3pt}

\newtheorem{theorem}{Theorem}
\crefformat{theorem}{#2Theorem~#1#3}
\Crefformat{theorem}{#2Theorem~#1#3}

\newcommand{\newtheoremwithcrefformat}[2]{%
  \newtheorem{#1}[theorem]{#2}%
  \crefformat{#1}{##2\MakeUppercase#1~##1##3}%
  \Crefformat{#1}{##2\MakeUppercase#1~##1##3}%
}
\newcommand{\newseptheoremwithcrefformat}[2]{%
  \newtheorem{#1}{#2}%
  \crefformat{#1}{##2\MakeUppercase#1~##1##3}%
  \Crefformat{#1}{##2\MakeUppercase#1~##1##3}%
}

\newtheoremwithcrefformat{lemma}{Lemma}
\newtheoremwithcrefformat{proposition}{Proposition}
\newtheoremwithcrefformat{observation}{Observation}
\newtheoremwithcrefformat{conjecture}{Conjecture}
\newtheoremwithcrefformat{corollary}{Corollary}
\newseptheoremwithcrefformat{claim}{Claim}
\theorembodyfont{\upshape}
\newtheoremwithcrefformat{example}{Example}
\newtheoremwithcrefformat{remark}{Remark}
\newseptheoremwithcrefformat{definition}{Definition}

\theoremstyle{nonumberplain}
\theoremheaderfont{\scshape}
\theorembodyfont{\normalfont}
\theoremsymbol{\ensuremath{\square}}
\newtheorem{proof}{Proof}

\theoremsymbol{\ensuremath{\lrcorner}}
\newtheorem{clproof}{Proof}

\def\cqedsymbol{\ifmmode$\lrcorner$\else{\unskip\nobreak\hfil
\penalty50\hskip1em\null\nobreak\hfil$\lrcorner$
\parfillskip=0pt\finalhyphendemerits=0\endgraf}\fi}

\newcommand{\Oh}{\mathcal{O}}
\newcommand{\Pp}{\mathcal{P}}
\newcommand{\Qq}{\mathcal{Q}}
\newcommand{\Oof}{\mathcal{O}}
\newcommand{\Cc}{\mathscr{C}}
\newcommand{\Dd}{\mathscr{D}}
\newcommand{\Tt}{\mathcal{T}}

\newcommand{\tw}{\mathrm{tw}}

\newcommand{\N}{\mathbb{N}}
\newcommand{\R}{\mathbb{R}}

\renewcommand{\phi}{\varphi}
\renewcommand{\epsilon}{\varepsilon}

\newcommand{\Ll}{\mathcal{L}}
\newcommand{\Ff}{\mathcal{F}}
\newcommand{\Ss}{\mathcal{S}}
\newcommand{\prt}{\mathcal{P}}
\newcommand{\Rr}{\mathcal{R}}
\newcommand{\qrt}{\mathcal{Q}}

\newcounter{aux}

\newcommand{\fout}{f_{\mathrm{out}}}
\newcommand{\enc}{\mathrm{enc}}
\newcommand{\wh}[1]{\widehat{#1}}

\renewcommand{\geq}{\geqslant}
\renewcommand{\leq}{\leqslant}

\renewcommand{\setminus}{-}

\title{Polynomial bounds for centered colorings on \\proper minor-closed graph classes}
\author{\vspace{-3mm}Micha\l~Pilipczuk}
\address{University of Warsaw, Poland\\[1pt]
\texttt{michal.pilipczuk@mimuw.edu.pl}}
\ead{michal.pilipczuk@mimuw.edu.pl}
\author{Sebastian Siebertz}
\address{University of Bremen, Germany\\[1pt]
\texttt{siebertz@uni-bremen.de}}
\ead{siebertz@uni-bremen.de}
\tnotetext[t0]{A preliminary version of this work was presented at the {ACM-SIAM} Symposium on Discrete Algorithms, {SODA} 2019~\cite{PilipczukS19}. This work 
  was done while
  S. Siebertz was affiliated with the University of Warsaw and was supported by the National Science Centre of
  Poland via POLONEZ grant agreement UMO-2015/19/P/ST6/03998, which
  has received funding from the European Union's Horizon 2020 research
  and innovation programme (Marie Sk\l odowska-Curie grant agreement
  No.\ 665778).
}

\begin{document}

\begin{frontmatter}
\input{abstract}
\begin{textblock}{5}(11.13, 13.45)
\includegraphics[width=38px]{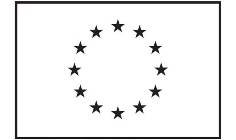}%
\end{textblock}

\end{frontmatter}

\input{intro}

\input{prelims}

\input{decompositions}

\input{planar}

\input{genus}

\input{almost-embeddable}

\input{minor-free}

\input{subgraph-iso}

\input{conclusions}

\paragraph*{Acknowledgements} The authors thank \'Edouard Bonnet for several useful comments, in particular spotting an issue in the proof of \cref{lem:color-coding}, and Zdenek Dvo\v{r}\'{a}k and Patrice Ossona de Mendez for insightful discussions on the problem.

\bibliographystyle{abbrv}
\bibliography{ref} 
\end{document}

%% file: abstract.tex
\begin{abstract}
  For $p\in \N$, a coloring $\lambda$ of the vertices of a graph $G$
  is {\em{$p$-centered}} if for every connected subgraph~$H$ of $G$,
  either $H$ receives more than $p$ colors under $\lambda$ or there is a color that
  appears exactly once in $H$.
  Centered colorings
  play an important role
  in the theory of sparse graph classes introduced by Ne\v{s}et\v{r}il
  and Ossona de Mendez~\cite{sparsity}, as they structurally
  characterize classes of {\em{bounded expansion}} --- one of the key
  sparsity notions in this theory.  More precisely, a
  class of graphs $\Cc$ has bounded expansion if and only if there is
  a function $f\colon \N\to \N$ such that every graph $G\in \Cc$ for
  every~$p\in\N$ admits a $p$-centered coloring with at most $f(p)$
  colors.  Unfortunately, known proofs of the existence of such
  colorings yield large upper bounds on the function~$f$ governing the
  number of colors needed, even for as simple classes as planar
  graphs.
  In this paper, we prove that every $K_t$-minor-free graph admits a
  $p$-centered coloring with~$\Oh(p^{\,g(t)})$ colors for some
  function $g$. In the special
  case that the graph is embeddable in a fixed surface $\Sigma$ we
  show that it admits a $p$-centered coloring with $\Oh(p^{19})$
  colors, with the degree of the polynomial independent of the genus
  of $\Sigma$.  This provides the first polynomial upper bounds on the
  number of colors needed in $p$-centered colorings of graphs drawn
  from proper minor-closed classes, which answers an open problem
  posed by Dvo\v{r}{\'a}k~\cite{DvorakP16}.
  As an algorithmic application, we use our main result to prove that if $\Cc$ is a fixed proper 
  minor-closed
  class of graphs, then given graphs $H$ and $G$, on $p$ and $n$
  vertices, respectively, where $G\in \Cc$, it can be decided whether
  $H$ is a subgraph of $G$ in time $2^{\Oh(p\log p)}\cdot n^{\Oh(1)}$
  and space $n^{\Oh(1)}$.
\end{abstract}


%% file: intro.tex
\section{Introduction}

\noindent Structural graph theory provides a wealth of tools that can be used in
the design of efficient algorithms for generally hard graph
problems. In particular, the algorithmic properties of classes of
graphs of bounded treewidth, of planar graphs, and more generally, of
classes which exclude a fixed minor have been studied extensively in
the literature. The celebrated structure theory developed
by Robertson and Seymour for graphs with
excluded minors had an immense influence on the design of efficient
algorithms. Ne\v{s}et\v{r}il and Ossona de Mendez introduced the even
more general concepts of \emph{bounded
  expansion}~\cite{nevsetvril2008grad} and \emph{nowhere
  denseness}~\cite{nevsetvril2011nowhere}, which offer abstract and
robust notions of sparseness in graphs, and which also lead to a rich
algorithmic theory.
Bounded expansion and nowhere dense graph classes were originally
defined by restricting the edge densities of bounded depth minors that
may occur in these classes; in particular, every class that excludes a fixed topological minor has bounded expansion.
In this work we are going to study \emph{$p$-centered colorings}, which may be used to give a structural characterization of bounded expansion and
nowhere dense classes, and which are particularly useful in the
algorithmic context.

\begin{definition}[\cite{nevsetvril2006tree}]
  Let $G$ be a graph, $p\in \N$, and let $C$ be a set of colors. A coloring
  $\lambda\colon V(G)\rightarrow C$ of the vertices of $G$ is
  called {\em{$p$-centered}} if for every connected subgraph~$H$ of
  $G$, either $H$ receives more than~$p$ colors or there is a color
  that appears exactly once in $H$ under $\lambda$.
\end{definition}

\begin{definition}
  For a function $f\colon \N\to \N$, we say that a graph class
  $\Cc$ admits {\em{$p$-centered colorings}} with $f(p)$ colors if for
  every $p\in \N$, every graph $G\in \Cc$ admits a $p$-centered
  coloring using at most $f(p)$ colors. If for the class $\Cc$ we can
  choose $f$ to be a polynomial, say of degree $d$, then we say
  that~$\Cc$ admits {\em{polynomial centered colorings}} of degree
  $d$.
\end{definition}

Ne\v{s}et\v{r}il and Ossona de
Mendez~\cite{nevsetvril2008grad} proved that classes of bounded expansion can be characterized by admitting centered
colorings with a bounded number of colors, as explained below.
\begin{theorem}[\cite{nevsetvril2008grad}]\label{thm:equiv-col-exp}
  A class $\Cc$ of graphs has bounded expansion if and only if there
  exists a function $f\colon \N\to \N$ such $\Cc$ admits $p$-centered
  colorings with $f(p)$ colors.
\end{theorem}

A similar characterization is known for nowhere dense classes as well, 
but this notion will not be directly relevant to the purpose of this work.
Note that \mbox{$1$-centered} colorings are exactly proper colorings of a
graph, thus centered colorings are a generalization of proper
colorings.  On the other hand, every $p$-centered coloring of a
graph~$G$ is also a {\em{treedepth-$p$ coloring}} of $G$, in the sense
that the union of every $i$ color classes, $i\leq p$, induces a
subgraph of~$G$ of treedepth at most $i$; see~\cite{nevsetvril2006tree}. 
Here, the {\em{treedepth}} of a graph is the minimum height of a rooted forest whose ancestor-descendant closure contains the graph;
this parameter is never smaller than the treewidth. 
Hence, a $p$-centered coloring of a graph $G$ can be understood as a decomposition of $V(G)$
into disjoint pieces, so that any subgraph induced by at most $p$
pieces is strongly structured --- it has treedepth at most $p$, so also treewidth at most $p$.

The inspiration of Ne\v{s}et\v{r}il and Ossona de Mendez for introducing low treedepth colorings in~\cite{nevsetvril2008grad} 
was a long line of research on {\em{low treewidth colorings}} in proper minor-closed classes (i.e., minor-closed classes excluding at least one minor).
It is a standard observation, underlying the classic Baker's approach, that if in a connected planar graph $G$ we fix a vertex $u$ and we color all the vertices 
according to the residue of their distance from $u$ modulo $p+1$, then the obtained coloring with $p+1$ colors has the following property: the union of any $p$ color classes induces a graph of treewidth~$\Oh(p)$.
As proved by Demaine et al.~\cite{demaine2005algorithmic} and by DeVos et al.~\cite{devos2004excluding}, such colorings with $p+1$ colors can be found for any proper minor-closed class of graphs.
Decompositions of this kind, together with similar statements for colorings of edges, 
are central in the design of approximation and parameterized algorithms in proper minor-closed graph classes,
see e.g.~\cite{demaine2005algorithmic,DemaineHK11,DemaineHM10,devos2004excluding} and the discussion therein.

Thus, low treedepth colorings offer a somewhat different view compared to low treewidth colorings: 
we obtain a stronger structure --- bounded treedepth instead of bounded treewidth --- at the cost of having significantly more colors --- some function of~$p$ instead of just $p+1$.
While admittedly not that useful for approximation algorithms, low treedepth colorings are a central 
algorithmic tool in the design of parameterized algorithms in classes of bounded expansion.
For instance, as observed in~\cite{nevsetvril2008grad}, using low treedepth colorings one can give a simple fpt algorithm for testing subgraph containment on classes of bounded expansion:
to check whether a graph $H$ on $p$ vertices is a subgraph of a large graph $G$, we compute a treedepth-$p$ coloring of $G$, say with $f(p)$ colors, and for every $p$-tuple of color classes
we use dynamic programming to verify whether $H$ is a subgraph of the graph induced by those color classes.
A much more involved generalization of this idea led to an fpt algorithm for testing any first-order definable property in any class of bounded expansion,
first given by Dvo\v{r}\'ak, Kr\'al' and Thomas~\cite{dvovrak2013testing} using different tools. 
We remark that proofs of this result using low treedepth colorings~\cite{lsd,grohe2011methods,circuit-mc} crucially use the fact that 
any $p$-tuple of color classes induce a graph of bounded treedepth, and not just bounded treewidth.

The running times of algorithms based on $p$-centered colorings strongly depend on
the number $f(p)$ of colors used. 
Unfortunately, the known approaches to constructing centered colorings produce a very large number of colors, typically exponential in $p$.
As shown in recent experimental works~\cite{abs-1712-06690}, this is actually a major bottleneck for applicability of these techniques in practice.

The original proof of \cref{thm:equiv-col-exp} in~\cite{nevsetvril2008grad} gives a bound for $f(p)$ that is at least doubly
exponential in $p$ for general classes of bounded expansion.
Somewhat better bounds for proper minor-closed classes can be established via a connection to yet another
family of parameters, namely the \emph{weak coloring numbers}, introduced by
Kierstead and Yang~\cite{kierstead2003orderings}. We refrain from
giving formal definitions, as they are not directly relevant to our purposes here, but intuitively the weak $p$-coloring number of a graph $G$
measures reachability properties up to distance~$p$ in a linear vertex ordering of the graph $G$.  It was shown by
Zhu~\cite{zhu2009coloring} that the number of colors needed for a
$p$-centered coloring of a graph is bounded by its weak $2^{p-2}$-coloring
number. The weak $r$-coloring number of a graph $G$ is bounded by~$\Oof(r^3)$ if $G$ is planar and by $\Oof(r^{t-1})$ if~$G$ excludes~$K_t$ as a minor~\cite{van2017generalised}. Combining
the two results gives a bound of~$\Oof(2^{3p})$ colors needed for a
$p$-centered coloring on planar graphs and $\Oof(2^{(t-1)p})$ on graphs
which exclude~$K_t$ as a minor. 
To the best of the authors' knowledge, so far no bounds polynomial in $p$ were known even for the case of planar graphs.

Motivated by this state of the art,
Dvo\v{r}\'ak~\cite{DvorakP16} asked whether one could
obtain a polynomial bound on the number of colors needed for
$p$-centered colorings on proper minor-closed  graph classes.

\vspace{-4pt}
\paragraph*{Our results}
We answer the question of Dvo\v{r}\'ak in affirmative by proving the following theorems.

\begin{theorem}\label{thm:main-minor}
  Every proper minor-closed class $\Cc$ admits
  polynomial time computable polynomial centered colorings, of some degree depending on $\Cc$.
\end{theorem}

 \newcounter{genus}
 \setcounter{genus}{\value{theorem}}

\begin{theorem}\label{thm:main-genus}
For every surface $\Sigma$, the class of graphs embeddable in $\Sigma$ admits polynomial time computable polynomial centered colorings of degree $19$.
More precisely, if the Euler genus of $\Sigma$ is $g$, then the obtained $p$-centered coloring uses $\Oh(g^2p^3+p^{19})$ colors.
\end{theorem}

Observe that in case of surface-embedded graphs we obtain a degree independent of the genus, however for general proper minor-closed classes the degree depends on the class.

\vspace{-4pt}
\paragraph*{Our techniques} 
Our proof proceeds by establishing the result for larger and larger graph classes.

We first focus on graphs of bounded treewidth, where we prove that the class of graphs of treewidth at most~$k$ admits polynomial centered colorings of degree~$k$, i.e. with $\binom{p+k}{k}=\Oh(p^k)$ colors for a $p$-centered coloring.
The key to this result is the combinatorial core of the proof of Grohe et al.~\cite{GroheKRSS15} that every graph of treewidth at most $k$ has weak $p$-coloring number $\binom{p+k}{k}$.

We next move to the case of planar graphs, where for every planar graph~$G$ we construct a $p$-centered coloring of $G$ using $\Oh(p^{19})$ colors.
The idea is to first prove a structure theorem for planar graphs, which is of independent interest.
To state it, we first need a few definitions

A path $P$ in a graph $G$ is called a {\em{geodesic}} if it is a shortest path between its endpoints.
  A {\em{partition}} of a graph $G$ is any family $\prt$ of induced
  subgraphs of $G$ such that every vertex of $G$ is in exactly one of
  subgraphs from $\prt$.  For a partition $\prt$ of $G$, we define the
  {\em{quotient graph}} $G/\prt$ as follows: it has $\prt$ as the
  vertex set and two parts $X,Y\in \prt$ are adjacent in $G/\prt$ if
  and only if there exist $x\in X$ and $y\in Y$ that are adjacent in
  $G$.
The structure theorem then can be stated as follows.

 \newcounter{planar}
 \setcounter{planar}{\value{theorem}}
\begin{theorem}
\label{thm:planar-partition}
  For every planar graph $G$ there exists a partition $\prt$ of $G$
  such that~$\prt$ is a family of geodesics in~$G$ and $G/\prt$ has treewidth at most $8$.
  Moreover, such a partition $\prt$ of $G$ together with a tree
  decomposition of $G/\prt$ of width at most $8$ can be computed in
  time $\Oh(n^2)$.
\end{theorem}

The idea of using separators that consist of a constant number of geodesics is not new.
A classic result of Lipton and Tarjan~\cite{lipton1979separator} 
states that in every $n$-vertex planar graph
one can find two geodesics whose removal leaves components
of size at most~$2n/3$. By recursively applying this result, one
obtains a decomposition of logarithmic depth along geodesic separators, which has found 
many algorithmic applications, see e.g.\ the notion of $k$-path separable graphs of Abraham and Gavoille~\cite{abraham2006object}. 
However, there is a subtle difference between this decomposition and the decomposition given by \cref{thm:planar-partition}:
in \cref{thm:planar-partition} all the paths are geodesics in the {\em{whole}} graph $G$, 
while in the decomposition obtained as above the paths are ordered as $P_1,\ldots, P_\ell$ so that each~$P_i$ is only geodesic in the graph
$G-\bigcup_{j<i}P_j$. This difference turns out to be crucial in our proof.

Let us come back to the issue of finding a $p$-centered coloring of a planar graph $G$.
By applying the layering technique, we may assume that that $G$ has radius bounded by~$2p$.
Hence, every geodesic in the partition $\prt$ given by \cref{thm:planar-partition} has at most $4p+1$ vertices.
By the already established case of graphs of bounded treewidth, the quotient graph $G/\prt$ admits a $p$-centered coloring $\kappa$ with $\Oh(p^8)$ colors (this is later blown up to $\Oh(p^{19})$ by layering).
We can now assign every vertex a color consisting of the color under $\kappa$ of the geodesic that contains it, and its distance
from a fixed end of the geodesic. This resolves the planar case.

We next lift the result to graphs embeddable in a fixed surface.
Here, the idea is to cut the surface along a short cut-graph that can be decomposed into $\Oh(g)$ geodesics; 
a construction of such a cut-graph was given by Erickson and Har-Peled~\cite{EricksonH04}.
Then the case of embeddable graphs is generalized to nearly embeddable graphs using a technical construction inspired by the work of Grohe~\cite{grohe2003local}.
Finally, we lift the case of nearly embeddable graphs to graphs from a fixed proper minor-closed class using the structure theorem of Robertson and Seymour~\cite{robertson2003graph}.
Here, we observe that the (already proved) bounded treewidth case can be lifted to a proof that $p$-centered colorings can be conveniently combined along tree decompositions with small adhesions, that is,
where every two adjacent bags intersect only at a constant number of vertices.

\vspace{-4pt}
\paragraph*{Applications}
Finally, we show a concrete algorithmic application of our main result.
There is one aspect where having a treedepth decomposition of small height is more useful than having a tree decomposition of small width, namely space complexity.
Dynamic programming algorithms on tree decompositions typically use space exponential in the width of the decomposition,
and there are complexity-theorerical reasons to believe that without significant loss on time complexity, this cannot be avoided.
On the other hand, on treedepth decompositions one can design algorithms with polynomial space usage.
We invite the reader to the work of Pilipczuk and Wrochna~\cite{PilipczukW18} for an in-depth study of this phenomenon.
This feature of treedepth can be used to prove the following.

 \newcounter{subiso}
 \setcounter{subiso}{\value{theorem}}

\begin{theorem}\label{thm:si}
Let $\Cc$ be a proper minor-closed class. Then given graphs $H$ 
and~$G$, on $p$ and $n$ vertices, respectively, where $G\in \Cc$, 
it can be decided whether $H$ is a subgraph of~$G$ in time $2^{\Oh(p\log p)}\cdot n^{\Oh(1)}$ and space $n^{\Oh(1)}$.
\end{theorem}

The proof of \cref{thm:si} follows the same strategy as before: having computed a $p$-centered coloring with $p^{\Oh(1)}$ colors, 
we iterate over all the $p$-tuples of color classes, of which there are $2^{\Oh(p\log p)}$, and for each $p$-tuple we use an algorithm that solves the problem on graphs of treedepth at most $p$.
This algorithm can be implemented to work in time $2^{\Oh(p\log p)}\cdot n^{\Oh(1)}$ and use polynomial space. 
We remark that this is not a straightforward dynamic programming, in particular we use the color-coding technique of Alon et al.~\cite{AlonYZ95} to ensure the injectivity of the constructed subgraph embedding.

The subgraph containment problem in proper minor-closed classes has a large literature.
By applying the same technique on a treewidth-$p$ coloring with~$p+1$ colors and using a standard dynamic programming algorithm for graphs of treewidth $p$ one can obtain an algorithm 
with time and space complexity \mbox{$2^{\Oh(p\log p)}\cdot n^{\Oh(1)}$} working on any proper minor-closed class. This was first observed for planar graphs by Eppstein~\cite{Eppstein99},
and the running time in the planar case was subsequently improved by Dorn~\cite{Dorn10} to $2^{\Oh(p)}\cdot n$.
More generally, the running time can be improved to $2^{\Oh(p/\log p)}\cdot n^{\Oh(1)}$ for apex-minor-free classes and connected pattern graphs $H$, and even to $2^{\Oh(\sqrt{p}\log^2 p)}\cdot n^{\Oh(1)}$ 
under the additional assumption that $H$ has constant maximum degree~\cite{FominLMPPS16}. All the abovementioned algorithms use space exponential in $p$.
See also~\cite{BodlaenderNZ16} for lower bounds under ETH.
Thus, \cref{thm:si} offers a reduction of space complexity to polynomial at the cost of having a moderately worse time complexity than the best known.

%% file: prelims.tex
\section{Lifting constructions}\label{sec:prelims}

\vspace{-4pt}
\paragraph*{Preliminaries} 
All graphs in this paper are finite and simple, that is, without loops at vertices or multiple edges connecting the same pair of vertices.
They are also undirected unless explicitly stated.
We use the notation of Diestel's textbook~\cite{diestel2012graph} and refer to it for all undefined notation.

When we say that some object in a graph $G$ is {\em{polynomial-time computable}}, or {\em{ptime computable}} for brevity, we mean that there is a polynomial-time algorithm that given a graph $G$ 
computes such an object in polynomial time.
When we say that $G$ admits a ptime computable $p$-centered coloring, we mean that the algorithm computing the coloring
takes $p$ on input and runs in time $c\cdot n^{c}$ for some constant $c$, independent of~$p$.
This definition carries over to classes of graphs: class $\Cc$ admits ptime computable $p$-centered colorings if there is an algorithm as above working on every graph from $\Cc$.
Note that in particular, the constant $c$ may depend on $\Cc$, but may not depend on $p$ (given with the input).

%

We now give three constructions that enable us to lift the existence of polynomial centered colorings from simpler to more complicated graph classes.
The first one is based on the layering technique, and essentially states that to construct $p$-centered colorings with $p^{\Oh(1)}$ colors in a graph class it suffices to focus on connected graphs of radius at most $2p$.
The second states that having a partition of the graph into small pieces, we can lift centered colorings from the quotient graph to the original graph.
The third allows lifting centered colorings through tree decompositions with small adhesion (maximum size of an intersection of two adjacent bags).

\vspace{-4pt}
\subsection{Lifting through layering}

\noindent We first show that using classic layering one can reduce the problem of finding $p$-centered colorings to connected graphs of radius at most $2p$.

\begin{lemma}\label{lem:bounded-radius} Let $\Cc$ be a minor-closed
  class of graphs. Suppose that for some function $f\colon \N\to \N$
  the following condition holds: for every $p\in \N$ and connected
  graph $G\in \Cc$ of radius at most $2p$, the graph $G$ has a
  ptime computable $p$-centered coloring with $f(p)$ colors. 
  Then $\Cc$ admits a ptime computable $p$-centered 
  colorings with $(p+1)\cdot f(p)^2$ colors. 
\end{lemma}
\begin{proof}
  Fix $p\in \N$. For any graph $G\in \Cc$ we shall construct a
  $p$-centered coloring of $G$ using $(p+1)\cdot f(p)^2$ colors. We
  may assume that $G$ is connected, as otherwise we treat each
  connected component of $G$ separately and take the union of the
  obtained colorings. Note here that each connected component of $G$
  belongs to $\Cc$, because $\Cc$ is minor-closed.

  Fix any vertex $u$ of $G$ and partition $V(G)$ into layers
  $L_0,L_1,L_2,\ldots\subseteq V(G)$ according to the distance from
  $u$: layer $L_i$ comprises vertices exactly at distance~$i$ from~$u$. Thus, $L_0=\{u\}$, $\{L_0,L_1,L_2,\ldots\}$ forms a partition
  of $V(G)$, and every edge of $G$ connects two vertices from same or
  adjacent layers. Let $k$ be the largest integer such that layer
  $L_k$ is non-empty.

  For every $j\in \{0,1,\ldots,k\}$ divisible by $p$, consider the
  graph $G_j$ defined as follows: take the subgraph of $G$ induced by
  $L_0\cup L_1\cup \ldots\cup L_{j+2p-1}$ and, provided $j>0$,
  contract all vertices of $L_0\cup L_1\cup \ldots \cup L_{j-1}$ onto
  $u$; note that this is possible since
  $L_0\cup L_1\cup \ldots\cup L_{j-1}$ induces a connected subgraph of
  $G$. Note that $G_j$ is obtained from~$G$ by vertex removals and
  edge contractions, so $G_j$ is a minor of $G$; since $\Cc$ is
  minor-closed, we have $G_j\in \Cc$. Moreover, $G_j$ is connected and
  has radius at most~$2p$: this is straightforward for $j=0$, while
  for $j>0$ it can be easily seen that every vertex of $G_j$ is at
  distance at most~$2p$ from the vertex resulting from contracting
  $L_0\cup L_1\cup \ldots \cup L_{j-1}$. Finally, the vertex set of
  $G_j$ contains the $2p$ consecutive layers
  $L_j,L_{j+1},\ldots,L_{j+2p-1}$, plus one more vertex when $j>0$.
  Thus, for every $i\in \{0,1,\ldots,k\}$ and vertex $v\in L_i$ we
  have that $v\in V(G_{j-1})$ and $v\in V(G_j)$, where
  $j=p\cdot \lfloor i/p\rfloor$ is the largest integer divisible by
  $p$ not larger than $i$. Here, $G_{j-1}$ should be ignored if $j=0$.
  Clearly, the layers and the graphs $G_j$ are polynomial time computable. 

  Since each $G_j$ is a graph from $\Cc$ that is connected and has
  radius at most~$2p$, we may apply the assumed property of $\Cc$ to
  $G_j$ in order to compute in polynomial time 
  a $p$-centered coloring $\lambda_j$ of
  $G_j$ using $f(p)$ colors. We may assume that all colorings~$\lambda_j$ use the color set $\{1,\ldots,f(p)\}$. Now, define a
  coloring $\lambda$ of $G$ as follows: for $i\in \{0,1,\ldots,k\}$
  with $j=p\cdot \lfloor i/p\rfloor$, to each vertex $v\in L_i$ assign
  a color $\lambda(v)$ consisting of the following three of numbers:
  \[i\bmod (p+1)\qquad ; \qquad \lambda_j(v) \qquad ; \qquad 
  \lambda_{j-1}(v)\textrm{ if }j>0\textrm{, and }1\textrm{
    otherwise.}\]
  These three numbers are arranged into an ordered triple as follows:
  $i\bmod (p+1)$ is always the first coordinate, while $\lambda_j(v)$
  is on the second coordinate if $j$ is even and on the third
  coordinate if $j$ is odd. The value $\lambda_{j-1}(v)$ (or $1$ if
  $j=0$) is put on the remaining coordinate. The ordered triple
  defined in this manner is set as the color $\lambda(v)$. Observe
  that thus, $\lambda$ is a coloring of $G$ using the color set
  $\{0,1,\ldots,p\}\times \{1,\ldots,f(p)\}\times \{1,\ldots,f(p)\}$,
  which consists of $(p+1)\cdot f(p)^2$ colors. Clearly, $\lambda$ 
  is polynomial time computable from the layers and the colorings 
  $\lambda_i$. So it remains to prove that $\lambda$ is a 
  $p$-centered coloring~of~$G$.

  To this end, fix any connected subgraph $H$ of $G$. Let
  $I\subseteq \{0,1,\ldots,k\}$ be the set of those indices $i$, for
  which $V(H)\cap L_i\neq \emptyset$. Since $H$ is connected, we have
  that $I$ is an interval, i.e., $I=\{a,a+1,\ldots,b\}$ for some
  $0\leq a\leq b\leq k$.

  Suppose first that $b-a>p$. Then for each residue
  $r\in \{0,1,\ldots,p\}$ there is $i\equiv r\bmod p$ such that
  $i\in I$, hence there is a vertex of $H$ whose color under $\lambda$
  has~$r$ on the first coordinate. We infer that vertices of $H$
  receive more than $p$ different colors under $\lambda$.

  Suppose now that $b-a\leq p$, which means that
  $V(H)\subseteq L_a\cup L_{a+1}\cup \ldots\cup L_{a+p-1}$. Let
  $j=p\cdot \lfloor a/p\rfloor$ be the largest integer divisible by
  $p$ not larger than~$a$. Then $a-j<p$, hence
  $V(H)\subseteq L_{j}\cup L_{j+1}\cup \ldots\cup L_{j+2p-1}\subseteq
  V(G_j)$
  and $H$ is an induced subgraph of~$G_j$. Since $\lambda_j$ is a
  $p$-centered coloring of $G_j$ and $H$ is a connected subgraph of
  $G_j$, we infer that either $H$ receives more than $p$ colors under
  $\lambda_j$, or some color in $\lambda_j$ appears exactly once among
  vertices of $H$. Moving to the coloring $\lambda$, observe that for
  every vertex $v$ of $H$, the color $\lambda_j(v)$ appears either on
  the second or on the third coordinate of the color $\lambda(v)$,
  depending on whether~$j$ is even or odd. Consequently, for any two
  vertices $v,v'\in V(H)$ we have that
  $\lambda_j(v)\neq \lambda_j(v')$ implies
  $\lambda(v)\neq \lambda(v')$, and the above mentioned property of
  $H$ under the coloring $\lambda_j$ carries over to $H$ under the
  coloring $\lambda$. 
\end{proof}

Note that in \cref{lem:bounded-radius}, if $f(p)$ is a polynomial of
degree $d$, then $\Cc$ admits ptime computable polynomial centered colorings of degree
$2d+1$.

\subsection{Lifting through partitions}

\noindent The following lemma will be useful for lifting the existence of
centered colorings through partitions and quotient graphs.

\begin{lemma}\label{lem:prt-lift}
Let $p,q$ be positive integers. Suppose a graph $G$ has a partition $\prt$, computable in ptime for $p$ given on input, so~that
\begin{itemize}
\item $|V(A)|\leq q$ for each $A\in \prt$, and
\item the graph $G/\prt$ admits a ptime computable $p$-centered coloring with $f(p)$ colors, for some $f\colon \N\to \N$.
\end{itemize}
Then the graph $G$ has a ptime computable $p$-centered coloring with $q\cdot f(p)$ colors.
\end{lemma}
\begin{proof}
  Since each part of $\prt$ has at most $q$ vertices, we can compute a coloring $\kappa\colon V(G)\to C$ for a color set $C$ of size $q$ 
  so that within each part of $\prt$ all vertices receive pairwise different colors.
  By assumption, we can also compute in polynomial time a $p$-centered coloring $\lambda_0\colon \prt\to D$ of $G/\prt$ for a color set $D$ of size $f(p)$.
  Let $\lambda\colon V(G)\to D$ be a natural lift of $\lambda_0$ to $G$: for each $u\in V(G)$ we put $\lambda(u)=\lambda_0(A)$, where $A\in \prt$ is such that $u\in V(A)$.
  We now construct the product coloring $\rho\colon V(G)\to C\times D$ defined~as
  $$\rho(u)=(\kappa(u),\lambda(u))\qquad \textrm{for each }u\in V(G).$$

  Since $\rho$ uses $q\cdot f(p)$ colors, it suffices to verify that $\rho$ is $p$-centered.  
  Let $G'=G/\prt$.
  Take any connected subgraph $H$ of $G$. Let
  ${\cal X}\subseteq \prt$ be the set of those parts of
  $\prt$ that intersect $H$. Since $H$ is connected, the graph
  $G'[{\cal X}]$ is connected as well. We infer that either parts
  from ${\cal X}$ receive more than $p$ different colors in $\lambda_0$,
  or there is a part $A\in {\cal X}$ whose color is unique in
  ${\cal X}$ under~$\lambda_0$. In the first case, it follows immediately that $H$
  receives more than $p$ different colors in~$\rho$, as there are
  already~$p$ different second coordinates of the colors of vertices of~$H$. In the second case, each vertex of $A$ receives a different
  color under $\lambda$, and no other vertex of $H$ can share this
  color, because $A$ is colored uniquely among~${\cal X}$. It follows that
  every vertex of $V(A)\cap V(H)$ has a unique color under~$\rho$
  among vertices of $H$; since this intersection is non-empty, the
  claim follows. 
\end{proof}

We will often use the following combination of \cref{lem:bounded-radius} and \cref{lem:prt-lift}, 
where each part of the partition is a geodesic.

\begin{corollary}\label{lem:geodesics-lift}
  Suppose that a minor-closed class of graphs $\Cc$ has the following
  property: for every graph $G\in \Cc$ there exists a ptime
  computable partition
  $\prt_G$ of~$G$ into geodesics in~$G$ so that the class
  $\Dd=\{G/\prt_G\colon G\in \Cc\}$
  admits ptime computable $p$-centered colorings with 
  $f(p)$ colors, for some function
  \mbox{$f\colon \N\to \N$}. Then~$\Cc$ admits ptime computable 
  $p$-centered colorings with
  \[(p+1)(4p+1)^2\cdot f(p)^2\] colors. 
\end{corollary} 
\begin{proof}
  By \cref{lem:bounded-radius}, it suffices to show that for every
  $p\in \N$, every connected graph $G\in \Cc$ of radius at most $2p$
  has a polynomial time computable 
  $p$-centered coloring with $(4p+1)\cdot f(p)$ colors. By
  assumption, there is a ptime computable 
  partition $\prt_G$ of $G$ such that every
  $P\in \prt_G$ is a geodesic in $G$ and the graph $H=G/\prt_G$ admits
  a ptime computable 
  $p$-centered coloring with $f(p)$ colors. Observe that
  any geodesic in a graph of radius at most $2p$ has length at most
  $4p$, hence each geodesic $P\in \prt_G$ contains at most $4p+1$
  vertices. The claim follows by \cref{lem:prt-lift}.
\end{proof}

%% file: decompositions.tex
\newcommand{\bag}{\beta}
\newcommand{\adh}{\alpha}
\newcommand{\mrg}{\mu}
\newcommand{\Torso}{\Gamma}
\newcommand{\Sk}{S}

\subsection{Lifting through tree decompositions}\label{sec:tdec}

\noindent In this paper, it will be convenient to work with rooted tree decompositions.
That is, the shape of a tree decomposition will be a {\em{directed tree}} $T$:
an acyclic directed graph with one \emph{root} node having out-degree~$0$
and all other nodes having out-degree $1$.
This imposes standard parent/child relation in $T$, where the parent of a non-root node is its unique out-neighbor.

\begin{definition}
  A {\em{tree decomposition}} of a graph $G$ is a pair
  $\Tt=(T,\bag)$, where $T$ is a directed tree and
  $\bag\colon V(T)\to 2^{V(G)}$ is a mapping that assigns each node
  $x$ of $T$ its {\em{bag}} $\bag(x)\subseteq V(G)$ so that the
  following conditions are satisfied:
  \begin{enumerate}[label=(T\arabic*),ref=(T\arabic*)]
  \item\label{p:amoeba} For each $u\in V(G)$, the set
    $\{x\colon u\in \bag(x)\}$ is non-empty and induces a connected
    subtree of $T$.
  \item\label{p:edge} For every edge $uv\in E(G)$, there is  $x\in V(T)$ such
    that $\{u,v\}\subseteq \bag(x)$.
  \end{enumerate}
\end{definition}

Let $\Tt=(T,\bag)$ be a tree decomposition of $G$.
The {\em{width}} of $\Tt$ is the maximum bag size minus~$1$, i.e., $\max_{x\in V(T)} |\bag(x)|-1$. 
The {\em{treewidth}} of a graph $G$ is the minimum possible width of a tree decomposition of $G$. 
For a non-root node~$x$ with parent $y$, we define the \emph{adhesion set} of $x$ as $\adh(x)=\bag(x)\cap \bag(y)$. If $x$ is the root, then we set $\adh(x)=\emptyset$ by convention.
The {\em{adhesion}} of the tree decomposition $\Tt=(T,\bag)$ is the maximum size of an adhesion set in $\Tt$, i.e., $\max_{x\in V(T)}|\adh(x)|$.

The following lemma expresses how centered colorings can be combined along tree decompositions with small adhesion.

\begin{lemma}\label{lem:tree-decomps}
  Let $k$ be a fixed integer and $\Cc$ be a class of graphs that
  admits ptime computable 
  polynomial centered colorings of degree $d$.  Suppose a class
  of graphs~$\Dd$ has the following property: every graph $G\in \Dd$
  admits a ptime computable tree decomposition over $\Cc$ with adhesion
  at most $k$. Then $\Dd$ has ptime computable 
  polynomial centered colorings of degree~$d+k$.
\end{lemma}

Before proving \cref{lem:tree-decomps}, we collect several properties
of tree decompositions. Let $\Tt=(T,\bag)$ be a tree decomposition of $G$. We use the following notation whenever $\Tt$ is clear from the context.
\begin{enumerate}
  \item We have a natural ancestor/descendant relation in $T$: a node is a
descendant of all the nodes that appear on the unique path leading from it to the root.
Note that every node of $T$ is also its own ancestor and descendant.
We write $x\leq_T y$ if $x$ is an ancestor of $y$.
Then $\leq_T$ is a partial order on the nodes of $T$ with the root being 
the unique $\leq_T$-minimal element.
  \item The \emph{margin} of a node $x$ is the set $\mrg(x)=\bag(x)\setminus \adh(x)$. Recall here that $\adh(x)$ is the adhesion set of $x$.
  \item For every vertex $u$ of $G$, let $x(u)$ be the unique $\leq_T$-minimal node of $T$ with $u\in \bag(x)$. Note that this node is unique due to condition~\ref{p:amoeba}.
        We define a quasi-order $\leq_{\Tt}$ on the vertex set of $G$ as follows: $u\leq_{\Tt} v$ if and only if $x(u)\leq_T x(v)$.
  \item The \emph{torso} of a node $x$ is the graph $\Torso(x)$
    on vertex set $\bag(x)$ where two vertices $u,v\in \bag(x)$ are
    adjacent if and only if $uv\in E(G)$ or if there exists $y\neq x$
    such that $u,v\in \bag(y)$. Equivalently, $\Torso(x)$ is obtained from $G[\bag(x)]$ by turning the adhesion sets of $x$ and of all children of $x$ into cliques.
  \item We call $\Tt$ a tree-decomposition \emph{over} a class
    $\Cc$ of graphs if $\Torso(x)\in \Cc$ for every node $x$ of~$T$.
  \item The \emph{skeleton} of $G$ over $\Tt$ is the directed graph $\Sk$
    with vertex set $V(\Sk)=V(G)$ and arc set defined as follows:
    for each $x\in V(T)$, $u\in \mrg(x)$, and $v\in \adh(x)$, we put the arc $(u,v)$ into the arc set of $\Sk$.
  \end{enumerate}

Note that if $(u,v)$ is an arc in the skeleton $\Sk$, then in particular $v<_{\Tt} u$, equivalently $x(v)<_T x(u)$. This implies that the skeleton is always acyclic (i.e. it is a DAG).

\pagebreak
The following lemmas express well-known properties of tree decompositions.

\begin{lemma}\label{lem:edge-anc}
If $\Tt=(T,\bag)$ is a tree decomposition of a graph $G$ and $uv$ is an edge in $G$, then $x(u)$ is an ancestor of $x(v)$ or vice versa.
Consequently, $u\leq_{\Tt} v$ or $v\leq_{\Tt} u$.
\end{lemma}
\begin{proof}
Otherwise the sets of nodes whose bags contain $u$ and $v$, respectively, would be disjoint, which would be a contradiction with the existence of the edge~$uv$ by condition~\ref{p:edge}.
\end{proof}

\begin{lemma}\label{lem:margins-prt}
Let $\Tt=(T,\bag)$ be a tree decomposition of a graph $G$.
For every vertex $u$ of $G$, the node $x(u)$ is the unique node of $T$ whose margin contains $u$.
Consequently, $\{\mrg(x)\}_{x\in V(T)}$ is a partition of the vertex set of~$G$.
\end{lemma}
\begin{proof}
Vertex $u$ belongs to $\mrg(x)$ for some node $x$ if and only if $u\in \bag(x)$ and either $x$ is the root, or the parent $y$ of $x$ satisfies $u\notin \bag(y)$.
By condition~\ref{p:amoeba}, among nodes $x$ with $u\in \bag(x)$ there is exactly one satisfying the second condition, being $x(u)$.
\end{proof}

Note that by \cref{lem:margins-prt}, the margins of nodes of $T$ are exactly the classes of equivalence in the quasi-order $\leq_{\Tt}$ on $V(G)$.

\medskip

We start the proof of \cref{lem:tree-decomps} by observing some properties of the skeleton graph.
Fix $p\in \N$, a graph~$G$, a tree decomposition $\Tt=(T,\bag)$ of $G$ with adhesion at most $k$, and let $\Sk$ be the skeleton of $G$ over~$\Tt$.

First, we show that restricted reachability in $G$ implies reachability in the skeleton.

\setcounter{claim}{0}

\begin{lemma}\label{lem:path-skeleton}
Let $u,v$ be vertices of $G$ with $v<_{\Tt} u$ and let $P$ be a path in $G$ with endpoints $u$ and $v$ such that every vertex $w$ of $P$ apart from $v$ satisfies $v<_{\Tt} w$.
Then there exists a directed path $Q$ in $\Sk$ leading from~$u$ to~$v$ and satisfying $V(Q)\subseteq V(P)$.
\end{lemma}
\begin{proof}
We proceed by induction on the length of path $P$. 
Let $w$ be the first (closest to~$u$) vertex on $P$ satisfying $w<_{\Tt} u$; such $w$ exists because $v$ satisfies the condition.

\begin{claim}\label{cl:uw-arc}
$(u,w)$ is an arc in $\Sk$.
\end{claim}
\begin{clproof}
Let $w'$ be the predecessor of $w$ on $P$.
We argue by induction that every vertex~$t$ on the prefix of~$P$ between $u$ and $w'$ satisfies $u\leq_{\Tt} t$.
This holds trivially for $t=u$. Supposing it holds for some vertex~$t$, we argue that it holds also for the successor~$t'$ of~$t$ on the prefix.
Indeed, we have $t\leq_{\Tt} t'$ or $t'\leq_{\Tt} t$ by \cref{lem:edge-anc}, implying that either $u\leq_{\Tt} t'$ or $u>_{\Tt} t'$, 
but the latter case is excluded by the choice of $w$, proving the induction step.
In particular, we infer that $u\leq_{\Tt} w'$, implying $w<_{\Tt} u\leq_{\Tt} w'$, which means that $x(w)<_T x(u)\leq_T x(w')$.

Since $ww'$ is an edge in $G$, assertion $x(w)<_T x(w')$ together with conditions~\ref{p:amoeba} and~\ref{p:edge} imply that $w$ is contained in all the bags of nodes on the unique path in~$T$ between 
$x(w)$ and $x(w')$. Since $x(w)<_T x(u)\leq_T x(w')$, we infer that $w$ is contained in the bag of both $x(u)$ and of its parent, which means that $w\in \adh(x(u))$.
As $u\in \mrg(x(u))$ by \cref{lem:margins-prt}, $(u,w)$ is an arc in $\Sk$, as claimed.
\end{clproof}

Since $w$ lies on $P$, by assumption we have either $w=v$ or $v<_{\Tt} w$.
In the former case we are immediately done, as we take $Q$ to be the path consisting only of the arc $(u,w)=(u,v)$.
In the latter case we may apply the induction assumption to $w$ and $v$ connected by the suffix of $P$ from $w$ to $v$. 
This yields a path $Q'$, which may be extended to a suitable path $Q$ by adding the edge $(u,v)$ at the front.
\end{proof}

\cref{lem:path-skeleton} suggests studying reachability in the skeleton.
We next show that if the considered tree decomposition has small adhesion, then every vertex reaches only a small number of vertices via short paths in the skeleton.
The argument essentially boils down to the combinatorial core of the proof that graphs of treewidth $k$ have weak $p$-coloring number $\binom{p+k}{k}$~\cite{GroheKRSS15}.

\setcounter{claim}{0}

\begin{lemma}\label{lem:skeleton-count}
Consider any vertex $u$ of $G$ and let $p\in \N$. Then there exist at most $\binom{p+k}{k}$
vertices of $G$ that are reachable by a directed path of length at 
most~$p$ from $u$ in $\Sk$.
\end{lemma}
\begin{proof}
  We proceed by induction on $p+k$, with base case $p=1$ or $k=1$.
  When $p=1$, we observe that the out-degrees in the skeleton are bounded by $k$, so there can be at most $k+1=\binom{k+1}{k}$ vertices reachable from $u$ by a path of length at most~$1$.
  When $k=1$, the skeleton $\Sk$ is a directed forest and every vertex reaches at most $p+1=\binom{p+1}{1}$ vertices by paths of length at most $p$.

  We proceed to the induction step.
  Let $x=x(u)$.
  By definition, the set of out-neighbors of $u$ in~$\Sk$ is exactly the adhesion set $\adh(x)$.
  If $\adh(x)=\emptyset$, then~$u$ has no out-neighbors and there is nothing to prove, so assume otherwise.
  For each $w\in \adh(x)$ we have $x(w)<_T x$, hence nodes $\{x(w)\colon w\in \adh(x)\}$ are pairwise comparable in the order $\leq_T$.
  Let $y$ be the $\leq_T$-minimal element of $\{x(w)\colon w\in \adh(x)\}$. 
  Note that $y<_T x$. Fix any $a\in\adh(x)$ with $y=x(a)$.

  Let $R$ be the set of vertices reachable from $u$ by a directed path of length at most~$p$ in $\Sk$; we need to prove that $|R|\leq \binom{p+k}{k}$.
  Note that whenever some $v\in R$ can be reached from $u$ by a directed path $P$ of length at most $p$ in $\Sk$, then all vertices on $P$ are contained in $R$; this is witnessed by the prefixes of $P$.
  Since arcs in the skeleton $\Sk$ point always to a vertex that is strictly smaller in the quasi-order $\leq_{\Tt}$, we have that $v<_{\Tt} u$ for each $v\in R\setminus \{u\}$;
  this in particular implies $x(v)\leq_T x$ for all $v\in R$.
  As $y<_T x$, we may partition $R$ into two subsets~$R_1$ and $R_2$ as follows:
  $$R_1=(R\cap\{v\colon x(v)<_T y\})\cup \{a\} \quad\textrm{and}\quad R_2=(R\cap\{v\colon y\leq_T x(v)\leq_T x\})\setminus \{a\}.$$
  Note that in particular $a\in R_1$ and $u\in R_2$, and by the choice of $y$ we have that~$a$ is the only out-neighbor of $u$ in $R_1$. We now analyze the interaction between $R_1$ and~$R_2$.
  
  \begin{claim}\label{cl:R2}
  In $\Sk$ there is no arc with tail in $R_1$ and head in $R_2$.
  \end{claim}
  \begin{clproof}
  Arcs in $\Sk$ always point to a vertex that is strictly smaller in the quasi-order~$\leq_\Tt$, but for each $s\in R_2$ and $t\in R_1$ we have $x(t)\leq_T y\leq_T x(s)$.
  \end{clproof}
  
  \begin{claim}\label{cl:R1}
  For every arc $(s,t)$ of $\Sk$ with $s\in R_2$ and $t\in R_1$, either $t=a$ or $(a,t)$ is an arc in $\Sk$.
  \end{claim}
  \begin{clproof}
  Supposing $t\neq a$ we have $x(t)<_T y$.
  As $(s,t)$ is an arc of $\Sk$, we have $t\in \adh(x(s))$, which implies $t\in \bag(x(s))$.
  As $s\in R_2$, we have $y\leq_T x(s)$, hence $x(t)<_T y\leq_T x(s)$.
  By condition~\ref{p:amoeba}, $t\in \bag(x(s))$ entails that $t$ belongs to every bag on the unique path in $T$ between $x(t)$ and $x(s)$, so in particular $t\in \adh(y)$.
  As $x(a)=y$, we find that $a\in \mrg(y)$ and $t\in \adh(y)$, hence $(a,t)$ is an arc in~$\Sk$.
  \end{clproof}
  
  \cref{cl:R2} and \cref{cl:R1} together imply the following.
  
  \begin{claim}\label{cl:summarize}
  Every vertex $v\in R_2$ is reachable from $u$ in $\Sk$ by a directed path of length at most $p$ whose vertices are all contained in $R_2$.
  Every vertex $v\in R_1$ is reachable from $a$ in $\Sk$ by a directed path of length at most $p-1$ whose vertices are all contained in $R_1$.
  \end{claim}
  \begin{clproof}
  For the first statement, by \cref{cl:R2} every path $P$ witnessing $v\in R$ has to be entirely contained in~$R_2$.
  For the second statement, suppose $P$ is a path witnessing $v\in R$.
  Let $s$ be the last vertex on $P$ that belongs to $R_1$ and let $t$ be its successor on~$P$; these vertices exist due to $u\in R_2$ and $v\in R_1$.
  If $t=a$, then we may take the suffix of~$P$ from $t$ to $v$.
  Otherwise, by \cref{cl:R1} we have that $(a,t)$ is an arc in $\Sk$.
  It then suffices to augment the suffix of $P$ from $t$ to~$v$ by adding the arc $(a,t)$ at the front.
  Observe that the path obtained in this manner has length at most $p-1$,
  because $t\in R_1\setminus \{a\}$ entails that $t$ is not an out-neighbor of $u$, which means that the prefix of~$P$ from $u$ to $t$ has length at least $2$.
  \end{clproof}
  
  By combining the induction assumption with the second statement of~\cref{cl:summarize} we infer that
  $$|R_1|\leq \binom{p-1+k}{k}.$$
  Observe now that for the graph $G[R_2]$ we may construct a tree decomposition $\Tt'$ from~$\Tt$ by removing all vertices outside of $R_2$ from all the bags, 
  and moreover removing all nodes not lying on the path from $x$ to~$y$ in~$T$; this is because $x(v)$ lies on this path for each $v\in R_2$.
  It is easy to see that the skeleton of $G[R_2]$ over~$\Tt'$ is equal to $\Sk[R_2]$.
  Further, observe that since $(u,a)$ is an arc of $\Sk$, we have $a\in \adh(x)\subseteq \bag(x)$, 
  which together with $a\in \bag(y)$ implies that in $\Tt$, the vertex $a$ is contained in the bags of all the nodes $z$ satisfying $y\leq_T z\leq_T x$.
  Since $a\notin R_2$, this implies that the adhesion of~$\Tt'$ is at most $k-1$.
  Then combining the induction assumption with the first statement of~\cref{cl:summarize} yields
  $$|R_2|\leq \binom{p+k-1}{k-1}.$$
  All in all, we have
  $$|R|=|R_1|+|R_2|\leq \binom{p+k-1}{k}+\binom{p+k-1}{k-1}=\binom{p+k}{k};$$
  this concludes the induction step.
\end{proof}

We may now use \cref{lem:skeleton-count} to find a coloring of the skeleton using a small number of colors so that all pairs of vertices connected by short paths receive different colors.

\begin{lemma}\label{lem:col-H}
  There is a coloring $\kappa\colon V(\Sk)\rightarrow C$ with a
  color set $C$ of size $\binom{p+k}{k}$ such that:
  \begin{itemize}
  \item if there exists a
  directed path of length at most $p$ from $u$ to $v$ in $\Sk$, then
  $\kappa(u)\neq \kappa(v)$; and 
  \item for every
  node $x\in V(T)$, all vertices of $\mrg(x)$ receive the same color under~$\kappa$.
  \end{itemize} 
  Moreover, given $G,\Tt,p$ on input, such a coloring can be computed in polynomial time.
\end{lemma}
\begin{proof}
As observed before, the skeleton $\Sk$ can be computed in polynomial time. 
We first compute the $p$-transitive closure $\Sk^p$ of $\Sk$,
which is a directed graph on the same vertex set as $\Sk$ where we put an arc $(u,v)$ whenever there is a path in $\Sk$ from~$u$ to~$v$ of length at most $p$.
By \cref{lem:skeleton-count}, every vertex has out-degree at most $\binom{p+k}{k}-1$ in~$\Sk^p$; the $-1$ summand comes from counting the vertex itself among reachable ones in \cref{lem:skeleton-count}.
As $\Sk$ is acyclic, so is $\Sk^p$ as well.
Hence, we can compute any topological ordering of $\Sk^p$ and iterate along the ordering while coloring the vertices greedily using the color set $\{1,\ldots,\binom{p+k}{k}\}$.
Specifically, every vertex $u$ receives the smallest color that is not present among the out-neighbors of $u$, which were all colored in the previous iterations.
This yields a proper coloring $\kappa$ of the undirected graph underlying $\Sk^p$, which is also coloring of $\Sk$ with the sought properties.
Note here that all vertices residing in the same margin $\mrg(x)$ for some $x\in V(T)$ have exactly the same out-neighbors in $\Sk$, so they will receive the same color in the procedure.
\end{proof}

We are ready to prove \cref{lem:tree-decomps}. 

\begin{proof}[of \cref{lem:tree-decomps}]
  Fix any graph $G\in \Dd$ and let $\Tt$ be a tree decomposition of $G$ over $\Cc$ with adhesion at most $k$; by assumption, such $\Tt$ can be computed in polynomial time.
  Let $\Sk$ be the skeleton of $G$ over~$\Tt$ and let $\kappa\colon V(G)\to C$ be the coloring of $\Sk$ provided by \cref{lem:col-H}.
  As asserted, $|C|=\binom{p+k}{k}=\Oh(p^k)$ and~$\kappa$ can be computed in polynomial time.
  
  We now define a coloring $\lambda$ of $G$ as follows.
  By assumption, for every node~$x$ of~$T$, we may compute in polynomial time a $p$-centered coloring $\lambda_x$ of the torso~$\Torso(x)$, 
  where each coloring $\lambda_x$ uses the same color set~$D$ of size $\Oh(p^d)$.
  Then $\lambda\colon V(G)\to D$ is defined as 
  $$\lambda(u)=\lambda_{x(u)}(u)\qquad\textrm{for each }u\in V(G).$$
  In other words, we restrict each coloring $\lambda_x$ to $\mrg(x)$ and $\lambda$ is the union of all those restrictions.

  We finally define a coloring $\rho\colon V(G)\to C\times D$ as the product of $\kappa$ and $\lambda$, that is
  $$\rho(u)=(\kappa(u),\lambda(u))\qquad\textrm{for each }u\in V(G).$$

  Note that $\rho$ uses $\Oh(p^{d+k})$ colors, hence it suffices to prove that $\rho$ is a \mbox{$p$-centered} coloring of $G$.
  
  To this end, fix any connected subgraph $H$ of $G$.
  Since $H$ is connected, by conditions~\ref{p:amoeba} and~\ref{p:edge} 
  we infer that the set $\{x\colon \bag(x)\cap V(H)\neq \emptyset\}$ is connected in~$T$.
  Consequently, this set contains a unique $\leq_T$-minimal node; call it $z$.
  Note that $z\leq_T x(u)$ for each $u\in V(H)$.
  Since the bag of the parent of $z$ (provided it exists) is disjoint with $V(H)$, we infer that $V(H)\cap \bag(z)=V(H)\cap \mrg(z)$.
  Fix any vertex $v\in V(H)\cap \mrg(z)$ and let $\kappa(v)=c$.

  Suppose first that there exists a vertex $u\in V(H)$ with $x(u)\neq z$ and $\kappa(u)=c$.
  Let $P$ be any path in $H$ connecting $u$ and $v$ and let $v'$ be the first (closest to $u$) vertex on $P$ satisfying $x(v')=z$.
  Then we have $v,v'\in \mrg(z)$, so by the properties of $\kappa$ asserted by \cref{lem:col-H} we have that $\kappa(v')=\kappa(v)=c$.
  As $x(u)\neq z$, we have $u\neq v'$.
  
  Let $P'$ be the prefix of $P$ from $u$ to $v'$.
  By the choice of $z$ and of $v'$ we have that $v'<_\Tt w$ for each vertex~$w$ on $P'$ different from $v'$.
  Hence, we may apply Lemma~\ref{lem:path-skeleton} to~$u$ and $v'$ to infer that in $\Sk$ there is a directed path $Q$ leading from~$u$ to $v'$ and satisfying $V(Q)\subseteq V(P')\subseteq V(H)$.
  By the properties of $\kappa$ asserted by \cref{lem:col-H}, every $p+1$ consecutive vertices on $Q$ receive pairwise different colors under $\kappa$.
  Since $u\neq v'$ and $\kappa(u)=\kappa(v')=c$, we conclude that $Q$ has length at least $p+2$ and among the first $p+1$ vertices of $Q$ there are $p+1$ different colors present.
  But $V(Q)\subseteq V(H)$, so $H$ receives more than~$p$ different colors under $\kappa$, hence it also receives more than~$p$ different colors under $\rho$ and we are done.
  
  We are left with the case when the vertices of $H$ that receive color $c$ under~$\kappa$ are exactly the vertices of $\mrg(z)$.
  Let $H'$ be the subgraph of $\Torso(z)$ induced by $V(H)\cap \mrg(z)$.
  We claim that $H'$ is connected.
  For this, take any $a,b\in V(H)\cap \mrg(z)$ and let $R$ be a path in $H$ connecting $a$ and $b$.
  By the properties of tree decompositions, for every maximal infix of $R$ lying outside of $V(H)\cap \mrg(z)$, the two vertices immediately preceding and immediately succeeding this infix on $R$ 
  have to belong to the same adhesion set $\adh(z')$ for some child $z'$ of $z$.
  As $\adh(z')$ is turned into a clique in $\Torso(z)$, we may shortcut this infix by using the edge connecting the two vertices.
  By performing this operation for infix of $R$ as above, we turn $R$ into a path $R'$ in $H'$ connecting $a$ and $b$.
  
  Since $H'$ is connected, the vertices of $V(H')=V(H)\cap \mrg(z)$ either receive more than $p$ different colors in $\lambda_z$, 
  or some color appears exactly once in $H'$ under $\lambda_z$.
  Recall that $\lambda$ and $\lambda_z$ coincide on $\mrg(z)$, so the above alternative holds for $\lambda$ as well.
  In the former case, $H'$ receives more than $p$ different colors in $\lambda$, implying the same for $H$ and $\rho$.
  In the latter case, let $v_0$ be a vertex whose color under $\lambda$ is unique in $H'$.
  Then the color of $v_0$ under $\rho$ is unique in $H$: the colors of all other vertices of $V(H)\cap \mrg(z)$ differ on the second coordinate,
  while the colors of all vertices of $V(H)\setminus \mrg(z)$ differ on the first coordinate.
  Since $H$ was chosen arbitrarily, this concludes the proof.
\end{proof}

With a slight modification of the proof of 
\cref{lem:tree-decomps} we obtain 
the following bound for graphs of bounded treewidth. 

\begin{lemma}\label{lem:bnd-tw}
For all $k,p\in \N$, the class of graphs of treewidth at most
$k$ admits ptime computable $p$-centered colorings with
$\binom{p+k}{k}$ colors. 
\end{lemma}
\begin{proof}
As shown by Bodlaender~\cite{bodlaender1996linear},
for every $k\in \N$, given a graph $G$ of treewidth at most~$k$
we may compute a tree decomposition~$\Tt=(T,\beta)$ of $G$ of width at most $k$ in linear time. We may assume that 
$|\mu(x)|\leq 1$ for all $x\in V(T)$, that is, each bag introduces
at most one new vertex: if there is $x\in V(T)$ with $\mu(x)=
\{v_1,\ldots, v_m\}$, then we can replace it by a path of 
new nodes $x_1,\ldots, x_m$ and $\beta(x_i)=\alpha(x)\cup \{v_1,\ldots, v_i\}$. It is easy to see that by iteratively applying these
modifications we obtain a tree decomposition with the desired property of the same width. 
  We may further assume that in $\Tt$ there are no two adjacent nodes with equal bags, as such two nodes can be contracted to one node with the same bag. Then $\Tt$ has adhesion at most $k$ and
  we may apply \cref{lem:col-H}. 
Let $\Sk$ be the skeleton of $G$ over~$\Tt$ and let $\kappa\colon V(G)\to C$ be the coloring of $\Sk$ provided by \cref{lem:col-H}.
As asserted, $|C|=\binom{p+k}{k}$ and~$\kappa$ can be 
computed in polynomial time. We claim that $\kappa$ is a 
$p$-centered coloring of $G$. 

  To see this, fix any connected subgraph $H$ of $G$.
  Since $H$ is connected, by conditions~\ref{p:amoeba} and~\ref{p:edge} 
  we infer that the set $\{x\colon \bag(x)\cap V(H)\neq \emptyset\}$ is connected in~$T$.
  Consequently, this set contains a unique $\leq_T$-minimal node; call it $z$. By our assumption that $|\mu(x)|=1$ for all $x\in V(T)$
  the quasi-order $\leq_T$ is a partial order and we have 
  $z <_T x(u)$ for each $u\in V(H)\setminus\{v\}$.
  Since the bag of the parent of $z$ (provided it exists) is disjoint with $V(H)$, we infer that $V(H)\cap \bag(z)=V(H)\cap \mrg(z)=\{v\}$
  for some vertex $v\in V(H)$. Note that this also implies
  $v<_\mathcal{T} u$ for all $u\in V(H)\setminus\{v\}$. Let 
  $\kappa(v)=c$. 
  
  Assume
  that there exists a vertex $u\in V(H)\setminus\{v\}$ with 
  $\kappa(u)=c$. Choose such~$u$ that has minimum distance
  to $v$. Let $P$ be a shortest path in $H$ connecting~$u$ and $v$. Then $P$ does not contain another
  vertex $w$ with $\kappa(w)=c$ (otherwise $u$ is not a vertex
  with minimum distance to $v$ and $\kappa(u)=c$). As 
  $v<_\mathcal{T} w$ for all $w\in V(P)\setminus\{v\}$ we may 
  apply \cref{lem:path-skeleton} to find a directed 
  path~$Q$ in $S$ leading from $u$ to $v$ and satisfying 
  $V(Q)\subseteq V(P)$. 
  By the properties of $\kappa$ asserted by \cref{lem:col-H}, every $p+1$ consecutive vertices on $Q$ receive pairwise different colors under $\kappa$.
  Since $u\neq v$ and $\kappa(u)=\kappa(v)=c$, we conclude that $Q$ has length at least $p+2$ and among the first $p+1$ vertices of $Q$ there are $p+1$ different colors present.
  But $V(Q)\subseteq V(H)$, so $H$ receives more than~$p$ different colors under $\kappa$.
  Since $H$ was chosen arbitrarily, this concludes the proof.
\end{proof}

%

%% file: planar.tex
\section{Planar graphs}\label{sec:planar}



%
\noindent In this section we establish the result for planar graphs. We first prove \cref{thm:planar-partition}, which we
repeat for convenience.
 
 \setcounter{claim}{0}

 \setcounter{aux}{\value{theorem}}
 \setcounter{theorem}{\value{planar}}
 
\begin{theorem}
  For every planar graph $G$ there exists a partition $\prt$ of $G$
  such that~$\prt$ is a family of geodesics in~$G$ and $G/\prt$ has treewidth at most $8$.
  Moreover, such a partition $\prt$ of $G$ together with a tree
  decomposition of $G/\prt$ of width at most $8$ can be computed in
  time $\Oh(n^2)$.
\end{theorem}

\setcounter{theorem}{\value{aux}}

\begin{proof}
  We provide a proof of the existential statement and at the end we
  briefly discuss how it can be turned into a suitable algorithm with
  quadratic time complexity.

  We may assume that $G$ is connected, for otherwise we may apply the
  claim to every connected component of $G$ separately and take the
  union of the obtained partitions. Let us fix any plane embedding of
  $G$. We also fix any triangulation~$G^+$ of $G$. That is, $G^+$ is a
  plane supergraph of~$G$ with $V(G^+)=V(G)$ whose embedding extends
  that of $G$, and every face of $G^+$ is a triangle. Let~$\fout$ be
  the outer face of $G^+$.

  In the following, by a {\em{cycle}} $C$ is $G^+$ we mean a simple
  cycle, i.e., a subgraph of $G$ consisting of a sequence of pairwise
  different vertices $(v_1,\ldots,v_k)$ and edges
  $v_1v_2,v_2v_3,\ldots,v_{k-1}v_k,v_kv_1$ connecting them in order.
  The embedding of such a cycle~$C$ splits the plane into two regions,
  one bounded and one unbounded. By the subgraph {\em{enclosed}} by
  $C$, denoted $\enc(C)$, we mean the subgraph of $G^+$ consisting of
  all vertices and edges embedded into the closure of the bounded
  region. Note that $C$ itself is a subgraph of $\enc(C)$. Moreover,
  $\enc(C)$ is a plane graph whose outer face is~$C$ and every
  non-outer face is a triangle.

  A cycle $C$ in $G^+$ shall be called {\em{tight}} if $C$ can be
  partitioned into paths
  $P_1,P_2,\ldots,P_k$ for some $k\leq 6$ so that each $P_i$
  ($i\in \{1,\ldots,k\}$) is a geodesic in $G$; in particular, all
  edges of $P_i$ belong to $E(G)$. Note that thus a tight cycle can
  contain at most $6$ edges from $E(G^+)-E(G)$, as each such edge must
  connect endpoints of two cyclically consecutive paths among
  $P_1,\ldots,P_k$. The crux of the proof lies in the following claim;
  see \cref{fig:planar-recurrence} for a visualization.

  \begin{claim}\label{cl:recursion} Let $C$ be a tight cycle, let
    $P_1,\ldots,P_k$ be a partition of $C$ witnessing its tightness,
    and let $H=\enc(C)$. Then there exists a partition $\qrt$ of $H$
    such that: \begin{enumerate}[label=(S\arabic*),ref=(S\arabic*)]
    \item\label{cnd:geodesics} $\qrt$ is a family of geodesics in $G$
      containing $P_1,\ldots,P_k$; and \item\label{cnd:tw} $H/\qrt$
      admits a rooted tree decomposition of width at most $8$ in which
      $P_1,\ldots,P_k$ belong to the root bag. \end{enumerate}
  \end{claim}

  \cref{thm:planar-partition} then follows by applying
  \cref{cl:recursion} to the outer face $\fout$, regarded as a tight
  cycle enclosing the whole graph $G^+$. Indeed, $\fout$ is a
  triangle, so partitioning it into three single-vertex geodesics
  witnesses that it is tight.

  We prove \cref{cl:recursion} by induction with respect to the number
  of bounded faces of $H=\enc(C)$. For the base of the induction, if
  $H$ has one bounded face, then $H=C$ is in fact a triangular,
  bounded face of $G^+$ and we may set $\qrt=\{P_1,\ldots,P_k\}$.

  Suppose then that $H$ has more than one face. Since $C$ is a simple
  cycle in~$G^+$, it has length at least $3$. Therefore, we may assume
  that $k\geq 3$, for otherwise we may arbitrarily split one of
  geodesics $P_i$ into two or three subpaths so that $C$ is
  partitioned into three geodesics, apply the reasoning, and at the
  end merge back the split geodesic in the obtained partition of $H$.
  Now, as $3\leq k\leq 6$, we may partition~$C$ into
  three paths $Q_1,Q_2,Q_3$ so that each $Q_j$ is either equal to some~$P_i$, or is equal to the concatenation of some~$P_i$ and $P_{i+1}$
  together with the edge of $C$ connecting them. Note that paths~$Q_j$
  are not necessarily geodesics.

\begin{figure}[!ht]
  \centering
  \def\svgwidth{0.65\columnwidth}
  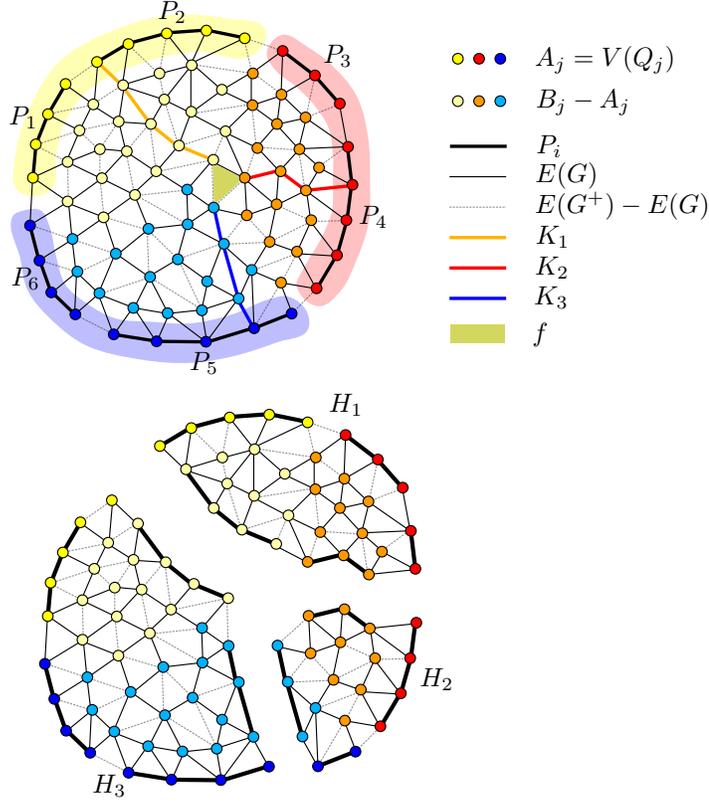 
  \caption{An example
                  situation in \cref{cl:recursion} and its proof. The
                  top panel depicts the construction of sets $A_j$ and
                  $B_j$, face $f$, and paths $K_j$. Note that $Q_1$ is
                  the concatenation of $P_1$ and $P_2$, $Q_2$ is the
                  concatenation of $P_3$ and $P_4$, and $Q_3$ is the
                  concatenation of $P_5$ and $P_6$. The bottom panel
                  depicts the graphs $H_j$, with cycles $C_j$
                  enclosing them, to which the claim is applied
                  inductively. Note that a partition of each cycle
                  $C_j$ into at most $6$ geodesics is depicted, which
                  witnesses that $C_j$ is
                  tight.}\label{fig:planar-recurrence}
\end{figure}

For $j\in \{1,2,3\}$, let $A_j=V(Q_j)$. Since $G$ is
connected, for every vertex $v$ of~$H$ there is some path in~$G$
connecting $v$ with $V(C)$. If we take the shortest such path, then it
is entirely contained in the graph~$H$. For $v\in V(H)$, let
$\pi(v)$ be the vertex of $V(C)$ that is the closest to~$v$ in $G$; in
case of ties, prefer a vertex belonging to $A_j$ with a smaller index
$j$, and among one $A_j$ break ties arbitrarily. Further, for each
$v\in V(H)$ fix $\Pi(v)$ to be any shortest path in $G$ connecting~$v$
with~$\pi(v)$; note that $\Pi(v)$ is a geodesic in $G$. For
$j\in \{1,2,3\}$, let $B_j\subseteq V(H)$ be the set of those vertices~$v$ of $H$ for which $\pi(v)\in A_j$; clearly, $\{B_1,B_2,B_3\}$ is a
partition of $V(H)$. Observe that $A_j\subseteq B_j$ and for every
vertex $v\in B_j$, the path~$\Pi(v)$ connects $v$ with~$A_j$, all its
vertices belong to $B_j$, and all its vertices apart from the endpoint
$\pi(v)$ do not lie on $C$.

Thus, we have partitioned the vertices of the disk-embedded graph $H$
into three parts $B_1,B_2,B_3$ so that $C$ --- the boundary of the
disk into which $H$ is embedded~--- is split into three nonempty
segments: one contained in $B_1$, one contained in $B_2$, and one
contained in $B_3$. All bounded faces of~$H$ are triangles. Hence, we
may apply Sperner's Lemma to $H$ to infer that there is a bounded face
$f$ of $H$ with vertices $v_1,v_2,v_3$ such that $v_j\in B_j$ for all
$j\in \{1,2,3\}$.

Let $K_j=\Pi(v_j)$, for $j\in \{1,2,3\}$. Observe that paths
$K_1,K_2,K_3$ are geodesics in $G$ and they are pairwise
vertex-disjoint, as each $K_j$ is entirely contained in~$B_j$.
Furthermore, the only vertex of $K_j$ that lies on $C$ is $\pi(v_j)$,
and moreover we have $\pi(v_j)\in A_j$. For $j\in \{1,2,3\}$ define
$C_j$ as the concatenation of: path $K_j$, edge $v_jv_{j+1}$ of the
face $f$, path $K_{j+1}$, and the subpath of $C$ between $\pi(v_j)$
and $\pi(v_{j+1})$ that is disjoint from $A_{j+2}$; here, indices
behave cyclically. From the asserted properties of $K_1,K_2,K_3$ it
follows that $C_j$ is a simple cycle, unless it degenerates to a
single edge~$v_jv_{j+1}$ traversed there and back in case $f$ shares
this edge with $C$. In the following we shall assume for simplicity
that all of $C_1,C_2,C_3$ are simple cycles; in case one of them
degenerates, it should be simply ignored in the analysis. Observe that
the disks bounded by $C_1,C_2,C_3$ are pairwise disjoint and if we
denote $H_j=\enc(C_j)$ for $j\in \{1,2,3\}$, then graphs~$H_j$ and~$H_{j+1}$ share only the path $K_{j+1}$. Moreover, each graph $H_j$
has strictly fewer bounded faces than $H$, since~$f$ is not a face of
any $H_j$.

Denote $L_j=K_j-\pi(v_j)$ for $j\in \{1,2,3\}$; in other words, $L_j$
is the path obtained from $K_j$ by removing the endpoint lying on $C$.
Observe that $L_j$ is a geodesic, unless it is empty in case
$v_j=\pi(v_j)$ lies on $C$. We now observe that each cycle $C_j$, for
$j\in \{1,2,3\}$, is tight. Indeed, $C_j$ can be partitioned into
geodesics $L_j$ and $L_{j+1}$ (provided they are not empty), a subpath
of $Q_j$, and a subgraph $Q_{j+1}$. By construction, path $Q_j$ can be
partitioned into one or two geodesics, so the same holds also for any
its subpath; similarly for any subpath of $Q_{j+1}$. We conclude that
$C_j$ can be partitioned into at most six geodesics: $L_j$ and
$L_{j+1}$ (provided they are not empty), one or two contained in
$Q_j$, and one or two contained in $Q_{j+1}$. This witnesses the
tightness of~$C_j$.

We now apply the induction hypothesis to each $C_j$ with the partition
witnessing its tightness as described above. This yields a suitable
partition $\qrt_j$ of~$H_j$ and tree decomposition $\Tt_j$ of
$H_j/\qrt_j$ of width at most $8$. Obtain a family~$\qrt_j'$ from the
partition~$\qrt_j$ as follows: for every path $R\in \qrt_j$ that is
contained in some geodesic $P_i$, for $i\in \{1,\ldots,k\}$, replace
$R$ with $P_i$. Note here that such paths $R$ have to be in the
partition of $C_j$ witnessing the tightness of $C_j$ and there can be
at most $4$ of them. Now let $$\qrt=\qrt_1'\cup \qrt_2'\cup \qrt_3'.$$
It follows readily from the construction that $\qrt$ is a partition of
$H$ into geodesics that contains all paths $P_i$ for
$i\in \{1,\ldots,k\}$; this yields condition~\ref{cnd:geodesics}.

For condition~\ref{cnd:tw}, consider a rooted tree decomposition $\Tt$
of $H/\qrt$ obtained as follows. Construct a root node $x$ with bag
$\beta(x)=\{P_1,\ldots,P_k\}\cup \{L_1,L_2,L_3\}$. Further, for each
$j\in \{1,2,3\}$ obtain $\Tt_j'$ from $\Tt_j$ by performing the same
replacement as before in every bag of $\Tt_j$: for every $R\in \qrt_j$
contained in a geodesic $P_i$, replace $R$ with $P_i$. Finally, attach
tree decompositions $\Tt_j'$ for $j\in \{1,2,3\}$ below $x$ by making
their roots into children of $x$. It can be easily seen that $\Tt$
obtained in this manner is indeed a tree decomposition of $H/\qrt$.
Moreover, we have $|\beta(x)|\leq 9$ and each decomposition $\Tt_j'$
has width at most $8$ by the induction hypothesis, so $\Tt$ has width
at most $8$ as well.

This finishes the proof of \cref{cl:recursion} and of the existential
part of \cref{thm:planar-partition}. The algorithmic statement
follows by turning the inductive proof into a recursive algorithm with
time complexity~$\Oh(n^2)$ in a straightforward way. Indeed, it is
easy to see that given $C$ as in \cref{cl:recursion}, the cycles
$C_1,C_2,C_3$ can be computed in linear time, and in the recursion we
investigate a linear number of recursive calls. 
\end{proof}

From \cref{thm:planar-partition}, \cref{lem:bnd-tw}, and
\cref{lem:geodesics-lift} we infer the result for planar
graphs.

\begin{theorem}\label{thm:planar} 
  The class of planar graphs admits ptime computable
  polynomial centered colorings of degree~$19$. 
\end{theorem}

%% file: planar-recurrence.pdf_tex
\begingroup%
  \makeatletter%
  \providecommand\color[2][]{%
    \errmessage{(Inkscape) Color is used for the text in Inkscape, but the package 'color.sty' is not loaded}%
    \renewcommand\color[2][]{}%
  }%
  \providecommand\transparent[1]{%
    \errmessage{(Inkscape) Transparency is used (non-zero) for the text in Inkscape, but the package 'transparent.sty' is not loaded}%
    \renewcommand\transparent[1]{}%
  }%
  \providecommand\rotatebox[2]{#2}%
  \ifx\svgwidth\undefined%
    \setlength{\unitlength}{429.37456665bp}%
    \ifx\svgscale\undefined%
      \relax%
    \else%
      \setlength{\unitlength}{\unitlength * \real{\svgscale}}%
    \fi%
  \else%
    \setlength{\unitlength}{\svgwidth}%
  \fi%
  \global\let\svgwidth\undefined%
  \global\let\svgscale\undefined%
  \makeatother%
  \begin{picture}(1,1.32076535)%
    \put(0,0){\includegraphics[width=\unitlength]{planar-recurrence.pdf}}%
    \put(-0.00123726,1.13022426){\color[rgb]{0,0,0}\makebox(0,0)[lb]{\smash{$P_1$}}}%
    \put(0.2502996,1.30944074){\color[rgb]{0,0,0}\makebox(0,0)[lb]{\smash{$P_2$}}}%
    \put(0.52911183,1.2350908){\color[rgb]{0,0,0}\makebox(0,0)[lb]{\smash{$P_3$}}}%
    \put(0.58901819,0.96222082){\color[rgb]{0,0,0}\makebox(0,0)[lb]{\smash{$P_4$}}}%
    \put(0.30399681,0.71697671){\color[rgb]{0,0,0}\makebox(0,0)[lb]{\smash{$P_5$}}}%
    \put(0.00262411,0.85863487){\color[rgb]{0,0,0}\makebox(0,0)[lb]{\smash{$P_6$}}}%
    \put(0.88653452,1.22951418){\color[rgb]{0,0,0}\makebox(0,0)[lb]{\smash{$A_j=V(Q_j)$}}}%
    \put(0.88792747,1.15995563){\color[rgb]{0,0,0}\makebox(0,0)[lb]{\smash{$B_j-A_j$}}}%
    \put(0.88861615,1.08282138){\color[rgb]{0,0,0}\makebox(0,0)[lb]{\smash{$P_i$}}}%
    \put(0.88714679,1.03323516){\color[rgb]{0,0,0}\makebox(0,0)[lb]{\smash{$E(G)$}}}%
    \put(0.88663016,0.98160915){\color[rgb]{0,0,0}\makebox(0,0)[lb]{\smash{$E(G^+)-E(G)$}}}%
    \put(0.88731885,0.92926808){\color[rgb]{0,0,0}\makebox(0,0)[lb]{\smash{$K_1$}}}%
    \put(0.88614151,0.87758925){\color[rgb]{0,0,0}\makebox(0,0)[lb]{\smash{$K_2$}}}%
    \put(0.88476414,0.82662565){\color[rgb]{0,0,0}\makebox(0,0)[lb]{\smash{$K_3$}}}%
    \put(0.88200941,0.76464267){\color[rgb]{0,0,0}\makebox(0,0)[lb]{\smash{$f$}}}%
    \put(0.13982204,0.00351529){\color[rgb]{0,0,0}\makebox(0,0)[lb]{\smash{$H_3$}}}%
    \put(0.53914835,0.64535909){\color[rgb]{0,0,0}\makebox(0,0)[lb]{\smash{$H_1$}}}%
    \put(0.69206107,0.181751){\color[rgb]{0,0,0}\makebox(0,0)[lb]{\smash{$H_2$}}}%
  \end{picture}%
\endgroup%

%% file: genus.tex
\section{Bounded genus graphs}\label{sec:genus}

\noindent In this section we lift the result to surface-embedded graphs.
By a {\em{surface}} we mean a compact, connected $2$-dimensional manifold $\Sigma$ without boundary.
An {\em{embedding}} of a graph $G$ in $\Sigma$ maps vertices of $G$ to distinct points in $\Sigma$ and edges of $G$ to pairwise non-crossing curves on $\Sigma$ connecting respective endpoints.
When we talk about a $\Sigma$-embedded graph, we implicitly identify the graph with its embedding in $\Sigma$.
For a $\Sigma$-embedded graph $G$, every connected component of $\Sigma-G$ is called a {\em{face}}.
The set of faces of $G$ is denoted by $F(G)$.
The embedding is {\em{proper}} if every face is homeomorphic to an open disk.

Recall that every surface $\Sigma$ has its {\em{Euler genus}} $g=g(\Sigma)$, which is an invariant for which the following holds: for every properly $\Sigma$-embedded connected graph~$G$, we have
$$|V(G)|-|E(G)|+|F(G)|=2-2g.$$

A subgraph $K$ of a properly $\Sigma$-embedded graph $G$ is called a {\em{cut-graph}} of $G$ if the topological space $\Sigma-K$ is homeomorphic to a disk.
The lift of our results from planar graphs to $\Sigma$-embeddable graphs is based on the following lemma, which was essentially proved by Erickson and Har-Peled in~\cite[Lemma 5.7]{EricksonH04}.
Since we need a slightly different phrasing, which puts focus on some properties of the construction that are implicit in~\cite{EricksonH04}, we provide our own proof.

\begin{lemma}\label{lem:genus-cutting}
Let $G$ be a connected graph properly embedded in a surface $\Sigma$ of Euler genus $g=g(\Sigma)$.
Then~$G$ contains a ptime computable 
cut-graph $K$ such that $K$ can be partitioned into at most $4g$ geodesics in~$G$.
\end{lemma}
\begin{proof}
\hspace{-0.3pt}As mentioned above, this result was essentially proved in~\cite[Lemma~5.7]{EricksonH04} and we follow closely the reasoning described there.
Our presentation is based on the presentation of~\cite[Lemma~10.1]{PilipczukPSL13}, which is another rephrasing of the same combinatorial fact.

In the proof we will use the dual multigraph $G^\star$.
Recall that the vertex set of $G^\star$ is the face set of $G$, and for every edge $e$ of $G$ we put in $G^\star$ the {\em{dual edge}} $e^\star$ connecting the two faces of $G$ incident to $e$.
We may define a proper $\Sigma$-embedding of $G^\star$ by mapping every face $f\in F(G)$ to any fixed point inside it, and every dual edge $e^\star=ff'$ to a suitable chosen curve connecting the points assigned to $f,f'$
that is contained in $f\cup f'\cup e$ and crosses $e$ once. Also, clearly~$G^\star$ is connected.

Fix any vertex $u\in V(G)$.
Let $T$ be the tree of a breadth-first search from $u$ in $G$.
Since $G$ is connected, $T$ is a tree on the same vertex set as $G$.
Moreover, for every $v\in V(G)$ the unique path in $T$ between~$u$ and~$v$ is a shortest $u$-$v$ path in $G$; in particular, it is a geodesic
in $G$. Let $M=\{e^\star\colon e\in E(T)\}\subseteq E(G^\star)$ be the set of edges dual to the edges of $T$.

Consider now the graph $G^\star-M$ and observe that it is connected, because $T$ is acyclic.
Let~$S$ be any spanning tree of $G^\star-M$.
Denote $N=E(S)$. 
Clearly, $M$ and $N$ are disjoint subsets of $E(G^\star)$ and we have $|M|=|V(G)|-1$ and $|N|=|V(G^\star)|-1=|F(G)|-1$.
Consequently, if we define $X^\star=E(G^\star)-(M\cup N)$, then
$$|X^\star|=|E(G^\star)|-|M|-|N|=|E(G)|-|V(G)|-|F(G)|+2=2g,$$
where $g$ is the Euler genus of $\Sigma$.
Denote $X=\{e\colon e^\star\in X^\star\}$.

Define a subgraph $K$ of $G$ by taking the union of the edges of $X$ and the $u$-$v$ paths in $T$, for all vertices~$v$ that are endpoints of edges in $X$.
Since $|X|=2g$, there are at most $4g$ such vertices $v$, hence $K$ is the union of the $2g$ edges of~$X$ and at most $4g$ paths $P_1,\ldots,P_k$ in $T$, where each $P_i$ is a geodesic in $G$.
By considering paths $Q_i=P_i-(V(P_1)\cup \ldots \cup V(P_{i-1}))$ and removing all $Q_i$s that turn out to be empty, we find that $K$ can be partitioned 
into at most~$4g$ geodesics
in $G$. Therefore, it remains to argue that $K$ is a cut-graph of $G$.

Consider first the graph $\widetilde{K}$ obtained from $T$ by adding all edges of $X$.
We argue that $\widetilde{K}$ is a cut-graph of~$G$.
To see this, fix any face $f\in F(G)$ and consider~$S$ as a tree rooted at $f$.
Note that $\Sigma-\widetilde{K}$ can be obtained from $f$ by iteratively gluing faces of $F(G)-\{f\}$ along edges dual to the edges of $S$, in a top-down manner on $S$.
Each face of $f$ is homeomorphic to a disk and gluing two disks along a common segment of their boundaries yields a disk.
Thus, throughout the above process we maintain the invariant that the topological space glued so far is homeomorphic to a disk, yielding at the end that $\Sigma-\widetilde{K}$ is homeomorphic to a disk.

Next, observe that $K$ can be obtained from $\widetilde{K}$ by iteratively removing vertices of degree $1$.
Note the following claim: if $H$ is a cut-graph of $G$ and $w$ is a vertex of degree $1$ in $H$, then $H'=H-w$ is also a cut-graph of $G$.
Indeed, $\Sigma-H'$ is obtained from the disk $\Sigma-H$ by gluing together the two segments of its boundary corresponding to the two sides of the unique edge of $H$ incident to~$w$.
These two segments share an endpoint and are glued in opposite directions, so $\Sigma-H'$ remains homeomorphic to a disk.
By applying this claim iteratively starting with $\widetilde{K}$, we infer that $K$ is indeed a cut-graph of $G$.
\end{proof}

\cref{lem:genus-cutting} can be now used to lift \cref{thm:planar-partition} to graphs embeddable into a surface of fixed genus.
The proof is a technical lift of the proof of \cref{thm:planar-partition}, precomposed with cutting the surface using the cut graph provided by \cref{lem:genus-cutting}.

\begin{theorem}\label{thm:genus-partition}
  Let $\Sigma$ be a surface of Euler genus $g$.
  Then for every graph $G$ that can be embedded into $\Sigma$,
  there is a ptime computable 
  partition $\prt$ of $G$ and a subset $\Qq\subseteq \prt$ with $|\Qq|\leq 16g$ such that
  $\prt$ is a family of geodesics in $G$ and $(G/\prt)-\Qq$ has treewidth at most $8$.
\end{theorem}
\begin{proof}[of \cref{thm:genus-partition}]
We assume that $G$ can be properly embedded in~$\Sigma$, since otherwise $G$ is embeddable in a surface of smaller genus and we can apply the reasoning for this surface.

By \cref{lem:genus-cutting}, we can compute a cut-graph $K$ of $G$ and its partition $\Qq$ into at most $4g$ geodesics in $G$.
Consider a graph $\wh{G}$ obtained by {\em{cutting $G$ along $K$}} as follows; see Figure~\ref{fig:cutting} for reference.
Starting from $G$, first duplicate every edge~$e$ of $K$ and embed the two copies of $e$ in a small neighborhood of the original embedding, thus creating a very thin face of length $2$ incident to both copies.
Next, for every vertex $u$ of~$K$ scan the edges of~$K$ incident to $u$ in the cyclic order around $u$ in the embedding; we note that we do not assume here that $\Sigma$ is orientable, 
as reversing the order yields the same construction.
For every pair of consecutive edges $e,e'$ in the order, create a copy of~$u$ and make it incident to one copy of $e$, one copy of~$e'$, and all the edges of $G$ lying between $e$ and $e'$ in the cyclic order, 
as in Figure~\ref{fig:cutting}. The original vertex~$u$ is removed and copies are shifted within a small neighborhood of the original placement of $u$ as in Figure~\ref{fig:cutting}.

Finally, remove from $\Sigma$ the space that is not contained in the (slightly shifted) faces of $G$; this space is depicted in grey in the last panel of Figure~\ref{fig:cutting}.
Since $K$ was a cut-graph, the obtained topological space is homeomorphic to a closed disc and the graph $\wh{G}$ is embedded into it.
Moreover, the boundary of $\wh{G}$ is a simple cycle $C_0$ that contains two copies of every edge of $K$ and as many copies of every vertex of $K$ as its degree in $K$.
Let $\pi\colon V(\wh{G})\cup E(\wh{G})\to V(G)\cup E(G)$ be the mapping sending every vertex and edge of $\wh{G}$ to its origin in $G$.

\begin{figure}[!ht]
  \centering
  \includegraphics[width=0.6\textwidth]{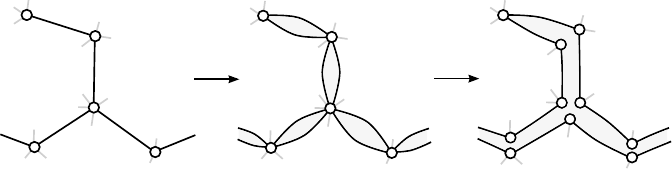}
  \caption{Cutting $G$ along $K$. The edges of $K$ are depicted in black, the remaining edges of $G$ are in grey. The grey area in the last panel is the part of the surface that gets removed.}\label{fig:cutting}
\end{figure}

The graph $\wh{G}$ is often called the {\em{polygonal schema}} of $\Sigma$, where $\Sigma$ is treated as a $2$-dimensional cell complex consisting of faces of $G$ glued along the edges of $G$.
We remark that in~\cite{EricksonH04}, Erickson and Har-Peled explain a different, equivalent construction of $\wh{G}$: take the faces of $G$, considered as closed discs, 
and instead of gluing them along all the edges of $G$ to obtain surface $\Sigma$, we glue them along the edges of $E(G)\setminus E(K)$ to obtain the disc $\Sigma-K$.
The graph obtained in this manner is $\wh{G}$.

Call an edge $e$ of $K$ {\em{delimiting}} if $e$ is not contained in any geodesic from~$\Qq$.
In the proof of \cref{lem:genus-cutting} we have argued that $K$ is a tree with $2g$ edges added, 

\noindent thus $|E(K)|=|V(K)|-1+2g$.
Since $|\Qq|\leq 4g$ and every geodesic $Q\in \Qq$ satisfies $|E(Q)|=|V(Q)|-1$, we have
\begin{align*}
|E(K)|-\sum_{Q\in \Qq} |E(Q)|&=|E(K)|-\sum_{Q\in \Qq} |V(Q)| + |\Qq|=|V(K)|-|E(K)|+|\Qq|\\
& =2g-1+|\Qq|\leq 6g.
\end{align*}
Hence, in $K$ there is at most $6g$ delimiting edges, so on $C_0$ there are at most $12g$ copies of delimiting edges.
Let $\Rr$ be the partition of $C_0$ into paths obtained by removing from $C_0$ all copies of all delimiting edges and taking all the obtained connected components as parts of $\Rr$.
Then $|\Rr|\leq 12g$ and every path in $\Rr$ is mapped under $\pi$ to a subpath of a path in~$\Qq$.

We now apply the reasoning leading to \cref{cl:recursion} from the proof of \cref{thm:planar-partition}.
First, we triangulate~$\wh{G}$, obtaining a triangulated, disc-embedded graph $\wh{G}^+$ with boundary being a simple cycle $C_0$.
Next, we redefine the notion of a tight cycle as follows: a cycle $C$ in $\wh{G}^+$ is {\em{tight}} if it admits a partition into paths in~$\wh{G}$, 
out of which all but at most $6$ are subpaths of different paths from $\Rr$, and the remaining at most $6$ paths are geodesics in $G$.
One can then readily verify that with the definition of tightness amended in this way, the inductive proof of \cref{cl:recursion} works just as before.
Here are the main differences:
\begin{itemize}
\item When defining paths $Q_1,Q_2,Q_3$, instead of requiring that they consist of at most two geodesics on $C$, 
      they are now a concatenation of an arbitrary number of paths from the partition witnessing tightness of $C$,
      however all but at most $2$ of these paths have to be subpaths of paths from $\Rr$.
\item When splitting the graph along paths $K_1,K_2,K_3$, we observe that $L_1,L_2,L_3$ (paths $K_1,K_2,K_3$ trimmed by removing the vertex lying on $C$) are geodesics in $G$, 
      as $K_1,K_2,K_3$ are defined as shortest paths in $\wh{G}$ from the face $f$ to the cycle $C$.
\end{itemize}
We leave verifying the straighforward details to the reader.

All in all, the above reasoning yields a partition $\wh{\prt}$ of $\wh{G}$ into paths such that $\Rr\subseteq \wh{\prt}$, all paths in $\wh{\prt}\setminus \Rr$ are geodesics in $G$, and $\wh{G}/\wh{\prt}$ 
admits a tree decomposition where every bag contains at most $9$ paths from $\wh{\prt}\setminus \Rr$.
Note now that since $\bigcup_{R\in \Rr} \pi(V(R)) = V(K)$, we have that $\wh{\prt}\setminus \Rr$ is a partition of $G-V(K)$ such that $(G-V(K))/(\wh{\prt}\setminus \Rr)$ has treewidth at most $8$.
Since $\Qq$ consists of geodesics in $G$ and $\bigcup_{Q\in \Qq} V(P) = V(K)$, we may output $\Qq$ and the partition $\prt=(\wh{\prt}\setminus \Rr)\cup \Qq$ of~$G$.
\end{proof}

We note that to prove \cref{thm:genus-partition}, one cannot just remove the cut-graph $K$, apply the planar case (\cref{thm:planar-partition}), and take the union of the resulting partition
and the partition of $K$ into $\Oh(g)$ geodesics. This is because the obtained geodesics would be geodesics in $G-V(K)$, and not in $G$.

From~\cref{thm:genus-partition} we may infer the result for graphs embedded into a fixed surface.

 \setcounter{aux}{\value{theorem}}
 \setcounter{theorem}{\value{genus}}
 
\begin{theorem}\label{thm:genus}
For every surface $\Sigma$, the class of graphs embeddable in $\Sigma$ admits polynomial time computable polynomial centered colorings of degree $19$.
More precisely, if the Euler genus of $\Sigma$ is $g$, then the obtained $p$-centered coloring uses $\Oh(g^2p^3+p^{19})$ colors.
\end{theorem}
\setcounter{theorem}{\value{aux}}
\begin{proof}
By Lemma~\ref{lem:bounded-radius}, it suffices to prove that every connected graph $G$ embeddable in~$\Sigma$ of radius at most~$2p$ admits a $p$-centered coloring with $\Oh(gp+p^9)$ colors.
Apply \cref{thm:genus-partition} to $G$, yielding a partition $\prt$ of $G$ into geodesics and $\Qq\subseteq \prt$ with $|\Qq|\leq 16g$ such that $(G/\prt)-\Qq$ has treewidth at most $8$.
Since a geodesic in a graph of radius at most $2p$ contains at most $4p+1$ vertices, we have that each geodesic in $\prt$ involves at most $4p+1$ vertices
and in particular the total number of vertices involved in geodesics from $\Qq$ is at most $16g(4p+1)=\Oh(gp)$.
By \cref{lem:bnd-tw}, the graph $(G/\prt)-\Qq$ admits a ptime computable $p$-centered coloring with $\Oh(p^8)$ colors.
Applying \cref{lem:prt-lift} to the graph $G'=G-\bigcup_{Q\in \Qq} V(Q)$ and its partition $\prt\setminus \Qq$, we obtain a ptime computable $p$-centered coloring of~$G'$ with $\Oh(p^{19})$ colors.
It now remains to extend this coloring to~$G$ by assigning each of the $\Oh(gp)$ vertices of $\bigcup_{Q\in \Qq} V(Q)$ a fresh, individual color.
\end{proof}

%% file: almost-embeddable.tex
\vspace{-3mm}
\section{Nearly embeddable graphs}\label{sec:nearly-embeddable}

\vspace{-1mm}
\noindent We now move to nearly embeddable graphs.
Roughly saying, a graph $G$ is {\em{$(a,q,w,g)$-nearly embeddable}} if it is embeddable into a surface of Euler genus $g$ modulo at most $a$ {\em{apices}} and at most $q$ {\em{vortices}} of width at most $w$ each. 
This is formalized next, following the definitional layer of Grohe~\cite{grohe2003local}.

For two graphs~$G$ and $H$, by $G\cup H$ we denote the graph with vertex set $V(G)\cup V(H)$ and edge set $E(G)\cup E(H)$; note that this makes sense also when~$G$ and $H$ share vertices or edges.
A {\em{path decomposition}} is a tree decomposition where the underlying tree is a path.
A {\em{boundaried surface}}~$\Sigma$ is a 2-dimensional compact manifold with boundary homeomorphic to $q$ copies of $S^1$, for some $q\in \N$, which shall be called the {\em{boundary cycles}}.
The {\em{Euler genus}} of such a boundaried surface $\Sigma$ is the Euler genus of the closed surface obtained from $\Sigma$ by gluing a disc along each boundary cycle.

\begin{definition}\label{def:nearly-embeddable}
A graph $G$ is {\em{$(a,q,w,g)$-nearly embeddable}} 
if there is a vertex subset $A$ with $|A|\leq a$, (possibly empty) subgraphs $G_0,\ldots, G_q$ of $G$, and a boundaried surface~$\Sigma$ of genus $g$ with $q$ boundary cycles $C^1,\ldots,C^q$ 
such that the following conditions hold:
\begin{itemize}
\item We have $G-A=G_0\cup G_1\cup \ldots\cup G_q$.
\item Graphs $G_0,G_1,\ldots,G_q$ have disjoint edge sets, and graphs $G_1,\ldots,G_q$ are pairwise disjoint.
\item Graph $G_0$ has an embedding into $\Sigma$ such that all vertices of $V(G_0)\cap V(G_i)$ are embedded on~$C^i$, for $1\leq i\leq q$.
\item For $1\leq i\leq q$, let $m_i=|V(G_0)\cap V(G_i)|$ and let $u^i_1,u^i_2,\ldots,u^i_{m_i}$ be the vertices of $V(G_0)\cap V(G_i)$ in the order of appearance on the cycle $C^i$.
Then $G_i$ has a path decomposition $\Tt_i=(T_i,\bag_i)$ of width at most $w$, where $T_i$ is a path $(x^i_1,\ldots,x^i_{m_i})$ and $u^i_j\in \bag(x^i_j)$ for all $1\leq j\leq m_i$.
\end{itemize}
\end{definition}

The vertices of $A$ in the above definition are called \emph{apices}, the subgraphs $G_1,\ldots, G_q$ are called \emph{vortices}, and $G_0$ is called the \emph{skeleton graph}.
We now lift the results to nearly embeddable graphs.


\setcounter{claim}{0}

\begin{theorem}\label{thm:nearly-embeddable}
For any fixed $a,q,w,g\in \N$, the class of $(a,q,w,g)$-nearly embeddable graphs admits polynomial centered colorings of degree $\Oh(gq\cdot w^{q})$.
Moreover, these centered colorings are ptime computable assuming the input graph is given together with a decomposition as in \cref{def:nearly-embeddable}.
\end{theorem}
\begin{proof}
Fix $p$ for which we need to find a $p$-centered coloring of the input graph~$G$. Since $q$ is a fixed constant, we may assume that $p\geq 4q$, for otherwise we replace~$p$ with $\max(p,4q)$ in the reasoning.
As in \cref{def:nearly-embeddable}, let $A$ be the apex set, $G_0$ be the skeleton graph, $G_1,\ldots ,G_q$ be vortices, and~$\Sigma$ be the target surface of the near-embedding.
Note that all of the above is given on input. Let $G'=G-A=G_0\cup G_1\cup \ldots \cup G_q$.

Assign to each apex $w\in A$ an individual color $c_w$ that will not be used for any other vertex.
Thus, by using at most $a$ additional colors we may focus on finding a $p$-centered coloring of~$G'$.

We first argue that without loss of generality we may assume that $G'$ is connected and has radius bounded linearly in $p$.
In previous sections we used \cref{lem:bounded-radius} for such purposes, but this time there is a technical issue: the class of $(a,q,w,g)$-nearly embeddable graphs is not necessarily minor-closed.
However, the proof of \cref{lem:bounded-radius} can be amended, as explained next.

\begin{claim}\label{cl:radius-reduction}
Without loss of generality we may assume that $G'$ is connected and has radius at most $2p+2q$.
\end{claim}
\begin{clproof}
%
Clearly, we may assume that $G'$ is connected, because we may treat every connected component separately and take the union of the obtained colorings.
Also, we argue that we may assume that each vortex $G_i$, for $1\leq i\leq q$, is connected and has radius at most $1$.
This can be easily achieved by adding edges to $G_i$ so that one of its vertices, say $v_i$, becomes universal, i.e., is adjacent to all the other vertices of $G_i$.
Note that $v_i$ can be added to every bag of the assumed path decomposition $\Tt_i$ of $G_i$, so the width of every vortex grows to $w+1$ at most.
Thus, after modification the obtained graph~$G'$ is $(0,q,w+1,g)$-nearly embeddable, but whether the width of vortices is $w$ or $w+1$ has no impact on the claimed asymptotic bound of $\Oh(gq\cdot w^q)$ on
the degree of polynomial centered colorings.

We now follow the steps of the proof of \cref{lem:bounded-radius}.
Construct the same layering structure: pick any vertex $u$ and partition the vertex set of $G'$ into layers $L_0,L_1,L_2,\ldots$ according to distances from~$u$; let $k$ be such that $L_k$ is the largest nonempty layer.
Since we assumed that each vortex $G_i$ has radius at most $1$, it is entirely contained in $3$ consecutive layers.
Next, for every $j\in \{0,1,\ldots,k\}$ divisible by $p$ we considered the graph $G'_j$ obtained from $G'$ by contracting all layers $L_t$ for $t<j$ onto $u$ and removing all layers $L_t$ for $t\geq j+2p$.
Let $I_j=\{j,j+1,\ldots,j+2p-1\}$ be the interval of indices of layers that are preserved by this construction.

We now extend the interval $I_j$ to achieve the following condition: every vortex~$G_i$ is either entirely contained or entirely disjoint with $\bigcup_{t\in I_j} L_t$.
This can be achieved by adding one or two indices to $I_j$, either from the lower or the higher end, 
as long as there exists a vortex that intersects some layer~$L_t$ with $t\in I_j$ and some other layer~$L_{t'}$ with $t'\notin I_j$. 
Note that here we use the assumption that every vortex is contained in at most $3$ consecutive layers.
Since there are $q$ vortices in total, after this operation the interval~$I_j$ consists of at most $2p+2q$ consecutive layers and 
$$I_j\subseteq \{j-2q,\ldots,j+2p-1+2q\}\subseteq \{j-p/2,\ldots,j+2p-1+p/2\},$$
where the last containment follows from the assumption $p\geq 4q$. 
Observe that this means that every index~$t$ is contained in at most $3$ intervals $I_j$, for $j\in \{0,1,\ldots,k\}$ divisible by $p$.

After this modification we proceed as in the proof of \cref{lem:bounded-radius}.
Namely, let~$G'_j$ be the graph obtained from $G'$ by contracting all layers below the lower end of $I_j$ onto $u$ and removing all layers above the higher end of $I_j$.
Observe that $G'_j$ is still $(0,q,w+1,g)$-nearly embeddable, as every vortex either got entirely contracted onto~$u$, or got entirely removed, or is entirely preserved intact.
Moreover, as in \cref{lem:bounded-radius} we have that $G'_j$ is connected and has radius at most~$2p+2q$.

We may apply the assumption to compute a $p$-centered coloring $\lambda_j$ of $G'_j$, for each relevant $j$, and superimpose the colorings $\lambda_j$ just as in \cref{lem:bounded-radius}.
Note that now, in the obtained coloring $\lambda$ the color of each vertex $v$ is a $4$-tuple instead of a $3$-tuple, because it consists of the index of $v$'s layer modulo $p+1$ and 
the colors of~$v$ under~$\lambda_j$ for those indices $j$ for which $v\in I_j$, and there are at most three such $j$s.
Thus, if each~$\lambda_j$ uses $p^{\Oh(gq\cdot w^q)}$ colors, then $\lambda$ uses $(p+1)\cdot p^{3\cdot \Oh(gq\cdot w^q)}=p^{\Oh(gq\cdot w^q)}$ colors. 
It is straightforward to see that the remainder of the reasoning of the proof of \cref{lem:bounded-radius} goes through without changes, yielding that $\lambda$ is $p$-centered.

\mbox{}
\end{clproof}

By Claim~\ref{cl:radius-reduction}, from now on we assume that $G'$ is connected and has radius at most $2p+2q$.
We may also assume that each vortex $G_i$ for $1\leq i\leq q$ is non-empty, otherwise we apply the reasoning for smaller $q$.
Note that by the connectedness this implies $m_i\geq 1$.
We now follow the construction of Proposition 3.8 of Grohe~\cite{grohe2003local}, who showed
that nearly embeddable graphs without apices have bounded local treewidth. 

Construct a graph $\wh{G}$ from $G_0$ as follows.
For each $i\in \{1,\ldots,q\}$, introduce a new vertex $z^i$ and add edges: $z^i u^i_j$ for all $1\leq j\leq m_i$, and $u^i_j u^i_{j+1}$ for all $1\leq j\leq m_i$, where $u^i_{m_i+1}=u^i_1$.
Let $\wh{G}_i$ be the subgraph of $\wh{G}$ consisting of vertices $\{z^i\}\cup \{u^i_j\colon 1\leq j\leq m_i\}$ and edges added above; note that~$\wh{G}_i$ has diameter~$2$. 
Let $\wh{\Sigma}$ be the closed surface of Euler genus $g$ obtained from $\Sigma$ by gluing a disk $D^i$ along the boundary cycle~$C^i$, for each $1\leq i\leq q$.
Then $\wh{G}$ is embeddable into $\wh{\Sigma}$, because each subgraph~$\wh{G}_i$ can be embdedded into $D^i$.
Also, $\wh{G}$ has not much larger diameter than~$G'$.

\begin{claim}\label{cl:hat-diam}
The graph $\wh{G}$ is connected and has diameter at most $8(p+q)+2$.
\end{claim}
\begin{clproof}
Take any two vertices $a$ and $b$ of $\wh{G}$, and suppose for a moment that $a,b\in V(G_0)$. 
Then $a$ and $b$ can be connected by a path $P$ in $G'$ of length at most $4(p+q)$, because the radius of $G'$ is at most $2(p+q)$. 
Observe that every maximal infix of~$P$ that traverses the edges of some vortex $G_i$, where $1\leq i\leq q$, can be replaced by a path of length $2$ in $H$ through the vertex $z^i$.
By performing such replacement for every such infix we obtain a path $P'$ in $H$ with the same endpoints and length at most $8(p+q)$.
To resolve the case when $a$ or $b$ is among the new vertices $z^1,\ldots,z^q$, it suffices to observe that each of them is adjacent to some vertex of $G_0$.
\end{clproof}

Now that $\wh{G}$ is embeddable into $\wh{\Sigma}$, we may apply \cref{thm:genus-partition} to 
construct a partition~$\wh{\Pp}$ of $\wh{G}$ such that~$\wh{\Pp}$ is a family of geodesics in~$\wh{G}$ and 
$\wh{G}/\wh{\Pp}$ has treewidth~$\Oh(g)$. 
Note that since $\wh{G}$ has diameter at most $8(p+q)+2$ by Claim~\ref{cl:hat-diam}, each geodesic in $\wh{\Pp}$ has at most $8(p+q)+3$ vertices.

We now observe that because graphs $\wh{G}_i$ have diameter at most $2$, geodesics in $\wh{P}$ have only small interaction with them.

\begin{claim}\label{cl:intersect-vortex}
For $1\leq i\leq q$, every path $P\in\wh{\Pp}$ contains at most $3$ vertices of~$\wh{G}_i$.
\end{claim}
\begin{clproof}
If $P$ contained more than $3$ vertices in $\wh{G}_i$, then two of them would be at distance more than $2$ on~$P$, but $\wh{G}_i$ has diameter at most $2$.
This would contradict the assumption that $P$ is a geodesic in $\wh{G}$.
\end{clproof}

Construct a partition $\Pp$ of $G'$ into paths as follows.
First, each vertex $v\in V(G_1)\cup \ldots \cup V(G_q)$, i.e. participating in any vortex, gets assigned to a single-vertex path consisting only of $v$.
The remaining vertices, those of $V(G_0)\setminus (V(G_1)\cup \ldots \cup V(G_q))$,
are partitioned into inclusion-wise maximal paths contained in paths in $\wh{P}$.
In other words, for every path $P\in \wh{\Pp}$ we remove all vertices of $V(\wh{G}_1)\cup \ldots \cup V(\wh{G}_q)$, thus splitting $P$ into several paths, and put all those paths into $\Pp$.
Note that since each part of $\wh{\Pp}$ contains at most $8(p+q)+3$ vertices, the same holds also for $\Pp$, even though paths in $\Pp$ are not necessarily geodesics~in~$G'$.

For the next, crucial step we shall need the following technical lemma of Grohe~\cite{grohe2003local}, which enables gluing tree decompositions along common interfaces.

\begin{lemma}[Lemma 2.2 of \cite{grohe2003local}]\label{lem:pw-tw}
  Let $G,H$ be graphs and let $(T,\bag)$ be a path decomposition
  of~$H$ of width at most $k$. Assume that $T$ is a path $(x_1,\ldots,x_m)$ for some $m\in \N$.
  Let $v_1,\ldots, v_m$ be a path in $G$ such that $v_i\in \bag(x_i)$ for $1\leq i\leq m$ and
  $V(G)\cap V(H)=\{v_1,\ldots, v_m\}$. Then $\tw(G\cup H)\leq (\tw(G)+1)(k+1)-1$.
\end{lemma}

We now claim the following.

\begin{claim}
The graph $G'/\prt$ has treewidth $\Oh(gq\cdot w^q)$.
\end{claim}
\begin{clproof}
Let $\wh{\prt}'$ be a partition of $\wh{G}$ obtained as follows.
Examine every path $P\in \wh{\prt}$ and partition it into subpaths: 
a single-vertex path for each $v\in V(\wh{G}_1)\cup \ldots\cup V(\wh{G}_q)$ traversed by $P$,
and maximal infixes of~$P$ consisting of vertices not in $V(\wh{G}_1)\cup \ldots\cup V(\wh{G}_q)$.
Add all the obtained subpaths to $\wh{\prt}'$, thus eventually obtaining a partition of~$\wh{G}$.
Note that by \cref{cl:intersect-vortex}, we add at most $6q+1$ subpaths of $P$ for each $P\in \wh{\prt}$.

As $\wh{G}/\wh{\prt}$ has treewidth $\Oh(g)$, it is easy to see that $\wh{G}/\wh{\prt}'$ has treewidth $\Oh(gq)$.
Indeed, we may take a tree decomposition of $\wh{G}/\wh{\prt}$ of width $\Oh(g)$ and replace every path $P\in \wh{\prt}$ with all its subpaths added to $\wh{\prt}'$ in every bag, thus
obtaining a tree decomposition of $\wh{G}/\wh{\prt}'$ of width $\Oh(gq)$.

Note that since each subgraph $\wh{G}_i$ of $\wh{G}$ contains the cycle $(u^i_1,\ldots,u^i_{m_i})$, 
in~$\wh{G}/\wh{\prt}'$ the single-vertex paths consisting of vertices $u^i_1,\ldots,u^i_{m_i}$ form a cycle in the same way.

Now it remains to observe that $G'/\prt$ is a subgraph of a graph that can be obtained from $\wh{G}/\wh{\prt'}$ by iteratively adding vortices $G_1,\ldots,G_q$, with path decompositions $\Tt_1,\ldots,\Tt_q$
of width at most $w$ (here, we implicitly identify vertices of vortices with single-vertex paths consisting of them). 
Noting that the prerequisites of \cref{lem:pw-tw} are satisfied, 
every such addition increases the treewidth from the current value, say $t$, to at most $(t+1)(w+1)-1$.
Since the treewidth of $\wh{G}/\wh{\prt'}$ is $\Oh(gq)$, it follows that the treewidth of $G'/\prt$ is $\Oh(gq\cdot w^q)$.
\end{clproof}

Since $G'/\prt$ has treewidth $\Oh(gq\cdot w^q)$, by \cref{lem:bnd-tw} it admits a ptime computable $p$-centered coloring with $p^{\Oh(qg\cdot w^q)}$ colors.
As each part of $\prt$ has at most $8(p+q)+3$ vertices, we may conclude by \cref{lem:prt-lift}.
\end{proof}

%% file: minor-free.tex
\section{Proper minor-closed classes}\label{sec:minor}

\noindent Our main result now follows easily by combining the structure 
theorem of Robertson and Seymour with the already prepared tools. 

\begin{theorem}[Robertson and Seymour~\cite{robertson2003graph}]\label{thm:structure-thm}
For every $t\in \N$ there exist $a,q,w,g,k\in \N$ such that 
every graph $G$ excluding $K_t$ as a minor admits a tree 
decomposition of adhesion at most $k$ over the class of
$(a,q,w,g)$-nearly embeddable graphs. 
\end{theorem}

Furthermore, it is known that a tree decomposition as stated in \Cref{thm:structure-thm}, together with decompositions of torsos witnessing their $(a,q,w,g)$-near embeddability,
can be computed in polynomial time, see~\cite{demaine2005algorithmic,GroheKR13,KawarabayashiW11}. 
Now \Cref{thm:main-minor} is an immediate consequence of this result combined with \cref{lem:tree-decomps} and \cref{thm:nearly-embeddable}, as every proper minor-closed class excludes some clique $K_t$ as a minor.

%% file: subgraph-iso.tex
\section{{\sc{Subgraph Isomorphism}} for graphs of bounded treedepth}\label{sec:iso}

In this section we prove \cref{thm:si}, but the main technical contribution is the proof of the
following lemma. 

\begin{lemma}\label{lem:SI-td}
Suppose we are given a graph $H$ on $p$ and a graph $G$ on $n$ vertices, together with a treedepth decomposition of $G$ of depth $d$.
Then it can be decided whether $H$ is a subgraph of $G$ in time $2^{\Oh(p\log d)}\cdot n^{\Oh(1)}$ and space $n^{\Oh(1)}$.
\end{lemma} 

We then apply the following connection of $p$-centered colorings
with low-treedepth colorings. 

\begin{definition}
Let $F$ be a rooted forest, i.e., a graph whose connected components
are rooted trees. The \emph{closure} of $F$, denoted $\mathrm{clos}(F)$
has as its vertex set the set $V(F)$ and it contains every edge $uv$ such
that $u,v$ are vertices of a tree $T$ of $F$ and $u\leq_T v$. The
\emph{height} of a tree $T$ is the maximal number of vertices on 
a root-leaf path of $T$. The \emph{treedepth} of a graph $G$ is the
minimum height of a forest~$F$ such that $G\subseteq \mathrm{clos}(F)$. 
Such a forest $F$ is called a \emph{treedepth decomposition} of~$G$. 
\end{definition}

\begin{proposition}[Ne\v{s}et\v{r}il and Ossona de Mendez~\cite{nevsetvril2006tree}]\label{prop:centered-lowtd}
Every $p$-centered coloring $\lambda\colon V(G)\to C$ of a graph~$G$ is also a treedepth-$p$ coloring of $G$ in the following sense: for any color subset $X\subseteq C$ with $|X|\leq p$,
the graph $G[\lambda^{-1}(X)]$ has treedepth at most $|X|$.
Furthermore, a treedepth decomposition of $G[\lambda^{-1}(X)]$
of depth at most $|X|$ can be computed in linear time. 
\end{proposition}

\cref{lem:SI-td} combined with the above can be now used to give a space-efficient fixed-parameter algorithm for {\sc{Subgraph Isomorphism}} on proper minor-closed classes, as explained in \cref{thm:si}.

 \setcounter{aux}{\value{theorem}}
 \setcounter{theorem}{\value{subiso}}
 
\begin{theorem}
Let $\Cc$ be a proper minor-closed class. Then given graphs $H$ and~$G$, on $p$ and $n$ vertices, respectively, where $G\in \Cc$, 
it can be decided whether $H$ is a subgraph of~$G$ in time $2^{\Oh(p\log p)}\cdot n^{\Oh(1)}$ and space $n^{\Oh(1)}$.
\end{theorem}
\setcounter{theorem}{\value{aux}}
\begin{proof}
By \cref{thm:main-minor}, in polynomial time we can compute a $p$-centered coloring~$\lambda$ of $G$ that uses at most $c\cdot p^c$ colors, 
where $c$ is a constant depending only on $\Cc$. Iterate through all color subsets of~$p$ colors and for each such subset $X$ consider the graph $G_X=G[\lambda^{-1}(X)]$.
Observe that since $H$ has~$p$ vertices, $H$ is subgraph of $G$ if and only if $H$ is a subgraph of $G_X$ for any such color subset $X$.
By \cref{prop:centered-lowtd}, $G_X$ has treedepth at most $p$ and a treedepth decomposition of $G_X$ of depth at most $p$ can be computed in linear time.
Hence, we may apply \cref{lem:SI-td} to verify whether $H$ is a subgraph of $G_X$ in time $2^{\Oh(p\log p)}\cdot n^{\Oh(1)}$ and space $n^{\Oh(1)}$.
Since there are $(c\cdot p^c)^p=2^{\Oh(p\log p)}$ color subsets $X$ to consider, and for each we apply an algorithm with time complexity $2^{\Oh(p\log p)}\cdot n^{\Oh(1)}$ and space complexity $n^{\Oh(1)}$,
the claimed complexity bounds follow.
\end{proof}

In the remainder of this section we give a polynomial-space fixed-parameter algorithm for the {\sc{Subgraph Isomorphism}} problem on graphs of bounded treedepth, i.e. we prove \cref{lem:SI-td}.
Recall that we are given graphs $H$ and~$G$, where $H$ has $p$ vertices and $G$ has $n$ vertices,
and moreover we are given a treedepth decomposition $F$ of $G$ of depth at most $d$.
The goal is to check whether~$H$ is a subgraph of $G$;
that is, whether there exists a {\em{subgraph embedding}} from $H$ to $G$, which is an injective mapping $\eta$ from $V(H)$ to $V(G)$ such that $uv\in E(H)$ entails $\eta(u)\eta(v)\in E(G)$.

We first use the color coding technique of Alon et al.~\cite{AlonYZ95} to reduce the problem to the colored variant, where in addition vertices of $G$ are labeled with vertices of $H$ 
and the sought subgraph embedding has to respect these labels.
More precisely, suppose we are given a mapping $\alpha\colon V(G)\to V(H)$.
We say that a subgraph embedding $\eta$ from $H$ to $G$
is {\em{compliant}} with $\alpha$ if $\alpha(\eta(u))=u$ for each $u\in V(H)$.
The following lemma encapsulates the application of color coding to our problem, and its proof repeats the argument from~\cite[Sections~4.2 and~4.3]{NaorSS95}.

\begin{lemma}\label{lem:color-coding}
There exists a family $\Ff$ consisting of $2^{\Oh(p\log p)}\cdot \log n$ functions from $V(G)$ to $V(H)$ so that the following condition holds:
for each injective function $\eta\colon V(H)\to V(G)$ there exists at least one function $\alpha\in \Ff$ such that $\alpha(\eta(u))=u$ for each $u\in V(H)$.
Moreover, $\Ff$ can be enumerated with polynomial delay and using polynomial working space.
\end{lemma}
\begin{proof}
For positive integers $p\leq q$, a family $\Ss$ of functions from $V(G)$ to $\{1,\ldots,q\}$ is called {\em{$p$-perfect}} if for every subset $W\subseteq V(G)$ of size $p$ there exists a function $f\in \Ss$
that is injective on $W$. As argued by Naor et al.~\cite[Lemma~2]{NaorSS95}, for $q=p^2$ there exists a $p$-perfect family $\Ss$ of size $p^{\Oh(1)}\cdot \log n$ that can be computed in polynomial time.
Further, let $\Ll$ be the family of all non-decreasing functions from $\{1,\ldots,p^2\}$ to $V(H)$ with the following property: for each $g\in \Ll$ and $w\in V(H)$, the set $g^{-1}(w)$ is an interval in $\{1,\ldots,p^2\}$. Note that $|\Ll|=2^{\Oh(p\log p)}$.
Define the family $\Ff$ as follows:
$$\Ff=\{g\circ f\ \colon\ f\in \Ss\textrm{ and }g\in \Ll\}.$$
Note that since $|\Ss|=p^{\Oh(1)}\cdot \log n$ and $|\Ll|=2^{\Oh(p\log p)}$, we indeed have that $|\Ff|\leq 2^{\Oh(p\log p)}\cdot \log n$.
Also, clearly~$\Ff$ can be enumerated with polynomial delay and using polynomial working space.

Finally, we verify that $\Ff$ satisfies the promised condition. Let $W=\eta(V(H))$. By the properties of $\Ss$, there exists $f\in \Ss$ that is injective on $W$.
Further, we can define a function $g\colon \{1,\ldots,p^2\}\to V(H)$ by first setting $g(f(\eta(w)))=w$ for each $w\in V(H)$, and extending this mapping to the remaining elements of $\{1,\ldots,p^2\}\setminus f(W)$ so that $g^{-1}(w)$ is an interval in $\{1,\ldots,p^2\}$ for each $w\in V(H)$. Thus $g\in \Ll$, and consequently if we define $\alpha=g\circ f$, then $\alpha \in \Ff$. It remains to note that $\alpha(\eta(w))=w$ for each $w\in V(H)$.
\end{proof}

By applying \cref{lem:color-coding}, to prove \cref{lem:SI-td} we may focus on the variant where we are additionally given a mapping $\alpha\colon V(G)\to V(H)$ and we seek a subgraph embedding
that is compliant with $\alpha$. Indeed, if we give an algorithm with the promised time and space complexity for this variant, then we may apply it for every function $\alpha$ from the family $\Ff$ enumerated
using \cref{lem:color-coding}. This adds a multiplicative factor of $2^{\Oh(p\log p)}\cdot n^{\Oh(1)}$ to the time complexity and an additive factor of $n^{\Oh(1)}$ to the space complexity, 
which is fine for the claimed complexity bounds. 

\medskip

\newcommand{\chld}{\mathsf{Chld}}
\newcommand{\tail}{\mathsf{Tail}}
\newcommand{\Val}{\mathsf{Val}}
\newcommand{\CC}{\mathsf{CC}}

Before we proceed to the algorithm, we introduce some notation.
Recall that~$F$ is the given treedepth decomposition $G$; that is, $F$ is a rooted forest of depth at most $d$ on the same vertex as~$G$ such that every edge of $G$ connects a vertex with its
ancestor in $F$. For $u\in V(G)$, we introduce the following notation:
\begin{itemize}
\item $\chld(u)$ is the set of children of $u$ in $F$;
\item $\tail(u)$ is the set of all strict ancestors of $u$ in $F$ (i.e., excluding $u$ itself);
\item $G_u$ is the subgraph of $G$ induced by the ancestors and descendants of $u$, including $u$ itself.
\end{itemize}

\pagebreak
A pair $(X,D)$ of disjoint subsets of vertices of $H$ is a {\em{chunk}} if
\begin{itemize}
\item $X$ is either empty or it induces a connected subgraph of $H$; and
\item in $H$ there is no edge with one endpoint in $X$ and second in $V(H)-(X\cup D)$.
\end{itemize}
Note that the second condition is equivalent to saying that that $N_H(X)$ is contained in $D$.
A {\em{subproblem}} is a quadruple $(u,X,D,\gamma)$, where
\begin{itemize}
\item $u$ is a vertex of $G$,
\item $(X,D)$ is a chunk, and
\item $\gamma$ is an injective function from $D$ to $\tail(u)$.
\end{itemize}
Note that the number of different subproblems is at most $3^p\cdot d^p\cdot n=2^{\Oh(p\log d)}\cdot n$.
The {\em{value}} of the subproblem $(u,X,D,\gamma)$, denoted $\Val(u,X,D,\gamma)$, is the boolean value of the following assertion:
\begin{center}
{\em{There exists a subgraph embedding $\eta$ from the graph $H[X\cup D]$ to the graph $G_u$\\ such that $\eta(X)\cap \tail(u)=\emptyset$ and $\eta$ restricted to $D$ is equal to $\gamma$.}}
\end{center}
An embedding $\eta$ satisfying the above will be called a {\em{solution}} to the subproblem $(u,X,D,\gamma)$.
Our algorithm will compute the values of subproblems in a recursive manner using the formula presented in the following lemma.
Here, $\gamma[w\to u]$ denotes~$\gamma$ extended by mapping $w$ to $u$, whereas for $Y\subseteq V(H)$ by $\CC(Y)$ we denote the family of vertex sets of the connected components of $H[Y]$.

\begin{lemma}\label{lem:SI-recursion}
Suppose $(u,X,D,\gamma)$ is a subproblem. Then the following assertions hold:
\begin{enumerate}[label=(\roman*),ref=(\roman*)]
\item\label{ass:leaf} If $u$ is a leaf of $F$, then $\Val(u,X,D,\gamma)$ is true if and only if either $X=\emptyset$ and $\gamma$ is a subgraph embedding from $H[D]$ to $G_u$,
or $X=\{w\}$ with $w=\alpha(u)$ and $\gamma[w\to u]$ is a subgraph embedding from $H[D\cup \{w\}]$ to $G_u$.
\item\label{ass:node-not-im} If $u$ is not a leaf of $F$ and $\alpha(u)\notin X$, then
$$\Val(u,X,D,\gamma)=\bigvee_{v\in \chld(u)} \Val(v,X,D,\gamma).$$
\item\label{ass:node-im} If $u$ is not a leaf of $F$ and $w=\alpha(u)\in X$, then 
\begin{align*}
\Val(u,X,D,\gamma)= & \bigvee_{v\in \chld(u)} \Val(v,X,D,\gamma)\, \vee\, \\ 
& \hspace{0.9cm}\bigwedge_{Z\in \CC(X-\{w\})} \bigvee_{v\in \chld(u)} \Val(v,Z,D\cup \{w\},\gamma[w\to u]).
\end{align*}
\end{enumerate}
\end{lemma}
\begin{proof}
Assertion~\ref{ass:leaf} is straightforward.

For assertion~\ref{ass:node-not-im}, observe that a solution $\eta$ to the subproblem $(u,X,D,\gamma)$ cannot map any vertex of $X$ to~$u$, because only the vertex $w=\alpha(u)$ can be mapped to~$u$,
and $w$ does not belong to $X$ by assumption. Moreover, since $H[X]$ is connected (due to $(X,D)$ being a chunk), $\eta(X)$ has to be entirely contained in one subtree of $F$ rooted at a child of $u$.
It follows that every solution to the subproblem $(u,X,D,\gamma)$ is a solution to one of the subproblems $\Val(v,X,D,\gamma)$ for $v$ ranging over the children of $u$ in $F$, and conversely every solution
to any of these subproblems is trivially also a solution to $(u,X,D,\gamma)$. The formula follows.

For assertion~\ref{ass:node-im}, observe that every solution $\eta$ to the subproblem $(u,X,D,\gamma)$ either maps $w=\alpha(u)$ to $u$, or does not map any vertex to $u$.
In the latter case, the same reasoning as for assertion~\ref{ass:node-not-im} yields that $\eta$ is also a solution to one of subproblems $(v,X,D,\gamma)$ for $v$ ranging over the children of $u$;
this corresponds to the first part of the formula.
Consider now the former case, that is, suppose that indeed $\eta(w)=u$. Denote $D'=D\cup \{w\}$ and $\gamma'=\gamma[w\to u]$ for brevity.
Then, for every connected component $Z\in \CC(X-\{w\})$, $\eta$ has to map $Z$ entirely into one subtree of $F$ rooted at a child of $u$. 
Hence, for at least one child~$v$ of~$u$ we have that $\eta$ restricted to $Z\cup D'$ witnesses that $\Val(v,Z,D',\gamma')$ is true.
Since this holds for every $Z\in \CC(X-\{w\})$, we conclude that $\bigwedge_{Z\in \CC(X-\{w\})} \bigvee_{v\in \chld(u)} \Val(v,Z,D',\gamma')$ is true.
Conversely, supposing that $\bigwedge_{Z\in \CC(X-\{w\})} \bigvee_{v\in \chld(u)} \Val(v,Z,D',\gamma')$ is true,
for each $Z\in \CC(X-\{w\})$ we may find $v_Z\in \chld(u)$ and a solution $\eta_Z$ to the subproblem $\Val(v_Z,Z,D',\gamma')$. 
Observe that solutions~$\eta_Z$ match~$\gamma'$ on $D'$, hence we may consider their union; call it $\eta$. 
It is straightforward to see that $\eta$ is a subgraph embedding from $H[D\cup X]$ to $G_u$ that is compliant with $\alpha$ and extends $\gamma$, i.e., it is a solution to $(u,X,D,\gamma)$.
Here, the only non-trivial condition is injectivity, but this is ensured by the fact that each solution $\eta_Z$ is compliant with $\alpha$: 
each vertex $x\in X-\{w\}$, say $x\in Z$, is mapped by $\eta_Z$ to a vertex belonging to $\alpha^{-1}(x)$,
so no two vertices lying in different connected components $Z,Z'\in \CC(X-\{w\})$ can be mapped by $\eta_Z$ and $\eta_{Z'}$ to the same vertex of $G_u$.
The formula follows.
\end{proof}

\newcommand{\solve}{\mathsf{ComputeVal}}

\cref{lem:SI-recursion} suggests the following algorithm for our problem.
First, define a recursive procedure $\solve$ that given a subproblem $(u,X,D,\gamma)$ computes its value using the recursive formula provided by \cref{lem:SI-recursion}.
Then the algorithm proceeds as follows: for each connected component of~$H$, say with vertex set~$X$, verify whether there exists a root $r$ of $F$ for which $\Val(r,X,\emptyset,\emptyset)$ is true;
this is done by invoking $\solve(r,X,\emptyset,\emptyset)$ for each root $r$ of $F$.
Finally, conclude that there is a subgraph embedding from $H$ to $G$ compliant with $\alpha$ if and only if for each connected component of $H$ this verification was positive.
The same reasoning as for \cref{lem:SI-recursion}, assertion~\ref{ass:node-im}, proves that this algorithm is correct.
Further, throghout the algorithm we store only a stack consisting of at most $d$ frames of recursive calls to $\solve$, 
each of polynomial size, so the overall space complexity is polynomial in $n$.
It remains to argue that the time complexity is $2^{\Oh(p\log d)}\cdot n^{\Oh(1)}$.

To this end, we claim that $\solve$ is invoked on every subproblem $(u,X,D,\gamma)$ at most once.
As there are $2^{\Oh(p\log d)}\cdot n$ subproblems in total and the internal computation of $\solve$ for each of them takes polynomial time, the promised running time follows from this claim.
To see the claim, observe that if $\solve$ is invoked on a subproblem $(u,X,D,\gamma)$, then one of the following assertions holds:
\begin{itemize}
\item $u$ is a root of $F$ and $\solve(u,X,D,\gamma)$ is invoked directly in the main algorithm;
\item $u$ has a parent $v$ and the subproblem $(u,X,D,\gamma)$ uniquely defines the call to $\solve$ where $\solve(u,X,D,\gamma)$ was invoked: this call was to subproblem
$(v,X',D',\gamma')$ where 
\begin{itemize}
\item $D'=D-\gamma^{-1}(v)$ if $v\in \gamma(D)$ and $D'=D$ otherwise, 
\item $\gamma'=\gamma|_{D'}$, and
\item $X'$ is the vertex set of the unique connected component of $G-D'$ that contains $X\cup \gamma^{-1}(v)$, or $X'=\emptyset$ when $X\cup \gamma^{-1}(v)=\emptyset$.
\end{itemize}
\end{itemize}

With the above observation, the claim follows immediately: the subproblems on which $\solve$ is invoked in the main algorithm are pairwise different,
while every other subproblem on which $\solve$ is invoked has a uniquely defined parent in the recursion tree. This means that every subproblem is solved at most once and we are done.

%% file: conclusions.tex
\section{Conclusions}\label{sec:conclusion}

\noindent In this paper we gave the first polynomial upper bounds on the number of colors needed for $p$-centered colorings on proper minor-closed graph classes,
including the first such upper bounds for planar graphs. 
Following our work, D{\k{e}}bski et al.~\cite{dkebski2019improved}
improved the bounds for planar graphs to $\Oh(p^3\log p)$ 
and also provided a lower bound of $\Omega(p^2\log p)$. 
For graphs of treewidth $k$, D{\k{e}}bski et al.\ showed 
that our bound of $\binom{k+p}{k}$ is tight. Most surprisingly, 
they also showed that bounded degree graphs admit 
$p$-centered colorings with $\Oh(p)$ which allows to derive
also a polynomial bound on the number of colors for graphs 
excluding a fixed graph as a topological minor.

After the publication of the preliminary version of this work~\cite{PilipczukS19}, it quickly turned out that our auxiliary structural result about partioning a planar graph into geodesics --- \cref{thm:planar-partition} --- is indeed of independent interest. Dujmovi\'c{} et al.~\cite{DujmovicJMMUW19} used our approach to prove the {\em{Product Structure Theorem}} for planar graphs, which strenghtens the statement of \cref{thm:planar-partition} in that the geodesics can be chosen as vertical paths from a fixed BFS forest. The Product Structure Theorem has been subsequently used to solve several long-standing open problems about planar and bounded-genus graphs, including boundedness of the queue number~\cite{DujmovicJMMUW19} and of the non-repetitive chromatic number~\cite{DujmovicEJWW19}, and almost tight bounds for adjacency labelings~\cite{DujmovicEJGMM20}. We refer to the recent survey of Dvo\v{r}\'ak et al.~\cite{DvorkaHJLW20} for a broader discussion of the topic.

%% file: main.bbl
\begin{thebibliography}{10}

\bibitem{DvorakP16}
{\em Question posed at the Workshop on Algorithms, Logic and Structure, 12-14
  December, 2016, University of Warwick}.
\newblock https://warwick.ac.uk/fac/sci/maths/people/staff/
  daniel\_kral/alglogstr/ openproblems.pdf.

\bibitem{abraham2006object}
I.~Abraham and C.~Gavoille.
\newblock Object location using path separators.
\newblock In {\em PODC 2006}, pages 188--197. ACM, 2006.

\bibitem{AlonYZ95}
N.~Alon, R.~Yuster, and U.~Zwick.
\newblock Color-coding.
\newblock {\em J. {ACM}}, 42(4):844--856, 1995.

\bibitem{bodlaender1996linear}
H.~L. Bodlaender.
\newblock A linear-time algorithm for finding tree-decompositions of small
  treewidth.
\newblock {\em SIAM Journal on Computing}, 25(6):1305--1317, 1996.

\bibitem{BodlaenderNZ16}
H.~L. Bodlaender, J.~Nederlof, and T.~C. van~der Zanden.
\newblock Subexponential time algorithms for embedding {$H$}-minor free graphs.
\newblock In {\em {ICALP 2016}}, volume~55 of {\em LIPIcs}, pages 9:1--9:14.
  Schloss Dagstuhl --- Leibniz-Zentrum f\"ur Informatik, 2016.

\bibitem{dkebski2019improved}
M.~D{\k{e}}bski, S.~Felsner, P.~Micek, and F.~Schr{\"o}der.
\newblock Improved bounds for centered colorings.
\newblock In {\em {SODA 2020}}, pages 2212--2226. {SIAM}, 2020.

\bibitem{demaine2005algorithmic}
E.~Demaine, M.~Hajiaghayi, and K.~Kawarabayashi.
\newblock Algorithmic graph minor theory: Decomposition, approximation, and
  coloring.
\newblock In {\em FOCS 2005}, pages 637--646. IEEE, 2005.

\bibitem{DemaineHK11}
E.~D. Demaine, M.~Hajiaghayi, and K.~Kawarabayashi.
\newblock Contraction decomposition in {$H$}-minor-free graphs and algorithmic
  applications.
\newblock In {\em {STOC 2011}}, pages 441--450. {ACM}, 2011.

\bibitem{DemaineHM10}
E.~D. Demaine, M.~Hajiaghayi, and B.~Mohar.
\newblock Approximation algorithms via contraction decomposition.
\newblock {\em Combinatorica}, 30(5):533--552, 2010.

\bibitem{devos2004excluding}
M.~DeVos, G.~Ding, B.~Oporowski, D.~P. Sanders, B.~Reed, P.~Seymour, and
  D.~Vertigan.
\newblock Excluding any graph as a minor allows a low tree-width 2-coloring.
\newblock {\em Journal of Combinatorial Theory, Series B}, 91(1):25--41, 2004.

\bibitem{diestel2012graph}
R.~Diestel.
\newblock {\em Graph Theory, 4th Edition}, volume 173 of {\em Graduate {T}exts
  in {M}athematics}.
\newblock Springer, 2012.

\bibitem{Dorn10}
F.~Dorn.
\newblock Planar {S}ubgraph {I}somorphism revisited.
\newblock In {\em {STACS 2010}}, volume~5 of {\em LIPIcs}, pages 263--274.
  Schloss Dagstuhl --- Leibniz-Zentrum f\"ur Informatik, 2010.

\bibitem{DujmovicEJGMM20}
V.~Dujmovi\'c{}, L.~Esperet, G.~Joret, C.~Gavoille, P.~Micek, and P.~Morin.
\newblock Adjacency labelling for planar graphs (and beyond).
\newblock {\em CoRR}, abs/2003.04280, 2020.

\bibitem{DujmovicEJWW19}
V.~Dujmovi\'c{}, L.~Esperet, G.~Joret, B.~Walczak, and D.~R. Wood.
\newblock Planar graphs have bounded nonrepetitive chromatic number.
\newblock {\em CoRR}, abs/1904.05269, 2019.

\bibitem{DujmovicJMMUW19}
V.~Dujmovi\'c{}, G.~Joret, P.~Micek, P.~Morin, T.~Ueckerdt, and D.~R. Wood.
\newblock Planar graphs have bounded queue-number.
\newblock In {\em FOCS 2019}, pages 862--875. {IEEE} Computer Society, 2019.

\bibitem{dvovrak2013testing}
Z.~Dvo{\v{r}}{\'a}k, D.~Kr{\'a}l', and R.~Thomas.
\newblock Testing first-order properties for subclasses of sparse graphs.
\newblock {\em J. ACM}, 60(5):36, 2013.

\bibitem{DvorkaHJLW20}
Z.~Dvo\v{r}{\'{a}}k, T.~Huynh, G.~Joret, C.~Liu, and D.~R. Wood.
\newblock Notes on graph product structure theory.
\newblock {\em CoRR}, abs/2001.08860, 2020.

\bibitem{Eppstein99}
D.~Eppstein.
\newblock Subgraph isomorphism in planar graphs and related problems.
\newblock {\em J. Graph Algorithms Appl.}, 3(3), 1999.

\bibitem{EricksonH04}
J.~Erickson and S.~Har{-}Peled.
\newblock Optimally cutting a surface into a disk.
\newblock {\em Discrete {\&} Computational Geometry}, 31(1):37--59, 2004.

\bibitem{FominLMPPS16}
F.~V. Fomin, D.~Lokshtanov, D.~Marx, M.~Pilipczuk, M.~Pilipczuk, and
  S.~Saurabh.
\newblock Subexponential parameterized algorithms for planar and
  apex-minor-free graphs via low treewidth pattern covering.
\newblock In {\em {FOCS 2016}}, pages 515--524. {IEEE} Computer Society, 2016.

\bibitem{lsd}
J.~Gajarsk{\'{y}}, S.~Kreutzer, J.~Ne\v{s}et\v{r}il, P.~{Ossona de Mendez},
  M.~Pilipczuk, S.~Siebertz, and S.~Toru{\'{n}}czyk.
\newblock First-order interpretations of bounded expansion classes.
\newblock {\em ACM Transactions on Computational Logic}, 21(4), 2020.

\bibitem{grohe2003local}
M.~Grohe.
\newblock Local tree-width, excluded minors, and approximation algorithms.
\newblock {\em Combinatorica}, 23(4):613--632, 2003.

\bibitem{GroheKR13}
M.~Grohe, K.~Kawarabayashi, and B.~A. Reed.
\newblock A simple algorithm for the graph minor decomposition --- logic meets
  structural graph theory.
\newblock In {\em {SODA 2013}}, pages 414--431. {SIAM}, 2013.

\bibitem{grohe2011methods}
M.~Grohe and S.~Kreutzer.
\newblock Methods for algorithmic meta theorems.
\newblock {\em Model Theoretic Methods in Finite Combinatorics}, 558:181--206,
  2011.

\bibitem{GroheKRSS15}
M.~Grohe, S.~Kreutzer, R.~Rabinovich, S.~Siebertz, and K.~S. Stavropoulos.
\newblock Coloring and covering nowhere dense graphs.
\newblock {\em {SIAM} Journal on Discrete Mathematics}, 32(4):2467--2481, 2018.

\bibitem{KawarabayashiW11}
K.~Kawarabayashi and P.~Wollan.
\newblock A simpler algorithm and shorter proof for the graph minor
  decomposition.
\newblock In {\em {STOC 2011}}, pages 451--458. {ACM}, 2011.

\bibitem{kierstead2003orderings}
H.~A. Kierstead and D.~Yang.
\newblock Orderings on graphs and game coloring number.
\newblock {\em Order}, 20(3):255--264, 2003.

\bibitem{lipton1979separator}
R.~J. Lipton and R.~E. Tarjan.
\newblock A separator theorem for planar graphs.
\newblock {\em SIAM Journal on Applied Mathematics}, 36(2):177--189, 1979.

\bibitem{NaorSS95}
M.~Naor, L.~J. Schulman, and A.~Srinivasan.
\newblock Splitters and near-optimal derandomization.
\newblock In {\em FOCS 1995}, pages 182--191. {IEEE} Computer Society, 1995.

\bibitem{nevsetvril2006tree}
J.~Ne{\v{s}}et{\v{r}}il and P.~Ossona~de Mendez.
\newblock Tree-depth, subgraph coloring and homomorphism bounds.
\newblock {\em European Journal of Combinatorics}, 27(6):1022--1041, 2006.

\bibitem{nevsetvril2008grad}
J.~Ne{\v{s}}et{\v{r}}il and P.~Ossona~de Mendez.
\newblock Grad and classes with bounded expansion {I}. {D}ecompositions.
\newblock {\em European Journal of Combinatorics}, 29(3):760--776, 2008.

\bibitem{nevsetvril2011nowhere}
J.~Ne{\v{s}}et{\v{r}}il and P.~Ossona~de Mendez.
\newblock On nowhere dense graphs.
\newblock {\em European Journal of Combinatorics}, 32(4):600--617, 2011.

\bibitem{sparsity}
J.~Ne\v{s}et\v{r}il and P.~{Ossona de Mendez}.
\newblock {\em Sparsity --- {G}raphs, {S}tructures, and {A}lgorithms},
  volume~28 of {\em Algorithms and combinatorics}.
\newblock Springer, 2012.

\bibitem{abs-1712-06690}
M.~P. O'Brien and B.~D. Sullivan.
\newblock Experimental evaluation of counting subgraph isomorphisms in classes
  of bounded expansion.
\newblock {\em CoRR}, abs/1712.06690, 2017.

\bibitem{PilipczukPSL13}
M.~Pilipczuk, M.~Pilipczuk, P.~Sankowski, and E.~J. van Leeuwen.
\newblock Network sparsification for {S}teiner problems on planar and
  bounded-genus graphs.
\newblock {\em CoRR}, abs/1306.6593, 2013.
\newblock Version 1.

\bibitem{PilipczukS19}
M.~Pilipczuk and S.~Siebertz.
\newblock Polynomial bounds for centered colorings on proper minor-closed graph
  classes.
\newblock In {\em {SODA} 2019}, pages 1501--1520, 2019.

\bibitem{circuit-mc}
M.~Pilipczuk, S.~Siebertz, and S.~Toru{\'{n}}czyk.
\newblock Parameterized circuit complexity of model-checking on sparse
  structures.
\newblock In {\em {LICS} 2018}, pages 789--798, 2018.

\bibitem{PilipczukW18}
M.~Pilipczuk and M.~Wrochna.
\newblock On space efficiency of algorithms working on structural
  decompositions of graphs.
\newblock {\em ACM Trans. on Computation Theory (TOCT)}, 9(4):18:1--18:36,
  2018.

\bibitem{robertson2003graph}
N.~Robertson and P.~D. Seymour.
\newblock Graph minors. {XVI}. {E}xcluding a non-planar graph.
\newblock {\em Journal of Combinatorial Theory, Series B}, 89(1):43--76, 2003.

\bibitem{van2017generalised}
J.~van~den Heuvel, P.~Ossona~de Mendez, D.~Quiroz, R.~Rabinovich, and
  S.~Siebertz.
\newblock On the generalised colouring numbers of graphs that exclude a fixed
  minor.
\newblock {\em European Journal of Combinatorics}, 66:129--144, 2017.

\bibitem{zhu2009coloring}
X.~Zhu.
\newblock Colouring graphs with bounded generalized colouring number.
\newblock {\em Discrete Mathematics}, 309(18):5562--5568, 2009.

\end{thebibliography}
